\documentclass[a4paper,UKenglish,cleveref, autoref, thm-restate]{lipics-v2021}
\usepackage{amsmath, amssymb, amsthm, amsfonts}
\usepackage{tabularx}
\usepackage[dvipsnames]{xcolor}
\usepackage{hyperref}
\usepackage{enumitem}
\usepackage{algorithmicx}
\usepackage{algorithm}
\usepackage[noend]{algpseudocode}
\usepackage{booktabs}
\usepackage{tikz}
\usetikzlibrary{arrows.meta, positioning, quotes}

\usepackage{cite}
\usepackage{graphicx}
\usepackage{textcomp}
\usepackage{xcolor}

\title{Efficient Prime Paths Generation} 

\author{Jakub Zelek}{Jagiellonian University, Cracow, Poland }{jakub.zelek@doctoral.uj.edu.pl}{https://orcid.org/0000-0002-1825-0097}{}
\author{Jakub Ruszil}{Jagiellonian University, Cracow, Poland }{jakub.ruszil@uj.edu.pl}{https://orcid.org/0000-0002-1825-0097}{}
\author{Adam Roman}{Jagiellonian University, Cracow, Poland}{adam.roman@uj.edu.pl}{https://orcid.org/0000-0002-1825-0097}{}
\author{Artur Polański}{Jagiellonian University, Cracow, Poland }{artur.polanski@uj.edu.pl}{https://orcid.org/0000-0002-1825-0097}{}
\authorrunning{J. Ruszil, J. Zelek, A. Roman, A. Polański} 

\Copyright{J.Zelek, J. Ruszil, A. Roman, A. Polański} 



\category{} 

\relatedversion{} 

\supplement{\url{https://github.com/JakubZelek/Diblob/tree/prime-path-algorithm-comparision/notebooks/experiments/prime_paths_papier}}

\nolinenumbers 

\EventEditors{}
\EventNoEds{2}
\EventLongTitle{42nd Conference on Very Important Topics (CVIT 2016)}
\EventShortTitle{CVIT 2016}
\EventAcronym{CVIT}
\EventYear{2016}
\EventDate{December 24--27, 2016}
\EventLocation{Little Whinging, United Kingdom}
\EventLogo{}
\SeriesVolume{42}
\ArticleNo{23}

\begin{document}

\maketitle

\begin{abstract}
Prime path coverage is a powerful structural testing criterion, but generating all prime paths in a directed graph remains computationally challenging due to the potentially exponential number of them. Existing approaches typically rely on enumerating large sets of candidate paths and filtering them, leading to high computational and memory overhead.

In this paper, we present a new approach to prime path generation based on a structural characterization of prime paths in terms of strongly connected components. This characterization yields non-trivial necessary conditions for valid path endpoints and reduces the problem to constrained cycle enumeration in an augmented graph. As a result, we avoid explicitly enumerating all simple paths and instead generate only feasible candidates.

Building on this insight, we design a streaming algorithm that outputs prime paths incrementally, using a Johnson-style traversal as a subroutine within a significantly reduced search space. The algorithm exploits SCC boundary crossings as natural pruning checkpoints — discarding partial paths the moment they are detected to be backward extendable, eliminating entire subtrees of the search space during traversal rather than filtering completed paths post hoc.

We implement our method and evaluate it on a large dataset of real-world control-flow graphs extracted from open-source C++ and Python projects. The results demonstrate that our approach consistently outperforms existing methods, while maintaining stable inter-output delay in practice.

\end{abstract}
\section{Introduction} \label{secIntro}
Path-based coverage criteria, such as Prime Path Coverage (PPC), play an important role in white-box testing for several compelling reasons. First, they are strong criteria because they subsume many well-established and rigorous coverage criteria. For instance, PPC subsumes All-DU-Paths Coverage, Edge-Pair Coverage, and Complete Round-Trip Coverage \cite{Rechtberger2022, ammann2008}. As a result, satisfying PPC implicitly ensures satisfaction of these weaker criteria. Second, path-based criteria capture long dependency chains that are often missed by more localized coverage measures. Because prime paths are not extendable, they tend to represent meaningful execution flows, including business-critical sequences and extended control dependencies. This makes PPC particularly effective at revealing faults associated with complex interactions and long-range program behavior. Third, criteria such as PPC strike a practical balance between thoroughness and feasibility. While exhaustive path testing is infeasible for most real-world programs due to the potentially infinite number of execution paths, PPC focuses on a finite yet critical subset of paths, thereby providing strong fault-detection capability without incurring prohibitive testing costs. Finally, prime paths can be viewed as the fundamental building blocks of program execution. Analogous to the Prime Factorization Theorem (which states that every natural number greater than one can be uniquely expressed as a product of prime numbers), every execution path in a program can be seen as a composition of prime paths. In this sense, prime paths form the essential elements from which all program control-flow behaviors are constructed.

Prime paths, intuitively, are simple cycles (cycles without repeating vertices but the first) or simple paths (paths with no repeating vertices) in a digraph that cannot be extended without losing the property of being simple (i.e., adding any vertex to the start or to the end would repeat some vertex in the path). We start with the introductory, motivating example. 

\begin{figure}[ht]
  \centering
  \begin{tikzpicture}[
    node/.style={circle, draw=black, fill=white, inner sep=1.5pt, minimum size=6pt},
    arr/.style={-{Stealth[length=6pt, width=5pt]}, semithick}
]

\def\n{4} 

\node[node] (s1) {};

\foreach \i in {1,...,\n} {
    \node[node] (u\i) [above right of=s\i]{};
    \node[node] (v\i) [below right of=s\i]{};

    \node[node] (s\the\numexpr\i+1\relax) [right=1cm of s\i] {};

    \draw[arr] (s\i) to (u\i);
    \draw[arr] (s\i) to (v\i);

     \draw[arr] (u\i) to (s\the\numexpr\i+1\relax);
     \draw[arr] (v\i) to (s\the\numexpr\i+1\relax);
}

 \node[left=2pt] at (s1) {$s$};
 \node[right=2pt] at (s5) {$t$};

\end{tikzpicture}
  \caption{A motivating example graph $D_4$}
  \label{fig:intro_example}
\end{figure}
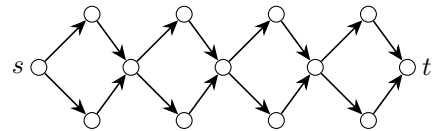

Observe that for the graph $D_4$ in Fig. \ref{fig:intro_example}, any of its prime paths must start in $s$ and end in $t$. Moreover, the graph has exactly $2^4 = 16$ prime paths. It is also easy to generalize that construction to $D_n$ to give $2^{\frac{n-1}{3}}$ prime paths, where $n$ is the number of vertices. Finally, it is worth noting that the resulting structures would be directed acyclic graphs (DAGs) that model the control-flow graphs (CFGs) of code composed of consecutive \texttt{if-else} blocks. This example motivates our research: since the number of prime paths can grow exponentially, even for DAGs, it is desirable to have a solution that computes prime paths one by one while reasonably using additional memory. In subsequent sections, we empirically confirm the motivation by running experiments.

The properties and advantages of PPC, as described above, make this coverage criterion the subject of numerous studies. Li \cite{Li2009}, Durelli \cite{Durelli2018}, and Monemi \cite{Monemi2020} experimentally compared PPC with mutation testing, Edge Coverage, Edge-Pair Coverage, and All-Uses Coverage, with PPC proving the most effective. Surprisingly, PPC generated only slightly more test coverage items than Edge-Pair Coverage.

Generating prime paths and the resulting optimal test sets that satisfy PPC is nontrivial, which is why many authors use AI-based approaches to do so. Silva \cite{Silva2020} emphasizes that there is a lack of tools for PPC and uses ML algorithms to predict the PPC based on a given test suite. Yadegari \cite{Yadegari2023} uses a self-adaptive algorithm that automatically designs test cases achieving PPC. Other authors use various evolutionary methods to achieve this goal, such as ant colony optimization and Markov chains \cite{Sayyari2015}, \cite{Bueno2002}, particle swarm optimization \cite{Bidgoli2017}, or genetic programming \cite{Goschen2022}. 

Dwarakanath \cite{Dwarakanath2014} first uses a classical prime path generation algorithm described in \cite{ammann2008} and then utilizes the minimum flow algorithm to generate an optimal test suite in $O(|V| \cdot |E|)$ time. Li \cite{Li2012} investigates the test generation problem as a shortest superstring problem, which is NP-complete, and presents some polynomial approximation algorithms. Dwarakanath and Jankiti \cite{Dwarakanath2014} exploit a max-flow/min-cut algorithm. Bure\v{s} \cite{Bures2019} generalizes the test generation problem in a more abstract setting where they take a model of the Program Under Test, a coverage criterion (e.g., prime path), the priority of each component in the Program Under Test, and an optimality criterion (e.g., minimum number of paths). Then, they generate a test set that meets an optimality criterion using search-based methods that exploit heuristic search techniques. 

Fazli \cite{Fazli2019} notes that most existing methods for generating prime/test paths have limited success in producing sets of all prime/test paths for structurally complex programs (e.g., those with high NPATH complexity). He proposes a time-efficient, vertex-based, compositional approach in which prime paths are computed for the SCCs of the CFG, and then composed with prime paths for the CFG whose vertices represent the SCCs. In \cite{Fazli2023}, Fazli parallelizes the approach from \cite{Fazli2019} along with a GPU implementation.

\subsection{Our contribution}

Our contributions are as follows:
\begin{itemize}
\item We provide a new characterization of prime paths in directed graphs in terms of strongly connected components (SCCs) and paths in the condensation graph. 
\item We establish an intractability result for the problem of enumerating non-extendable simple paths with polynomial delay, assuming a specific order of outputs.
\item We derive a set of non-trivial necessary conditions that any start or end vertex of a prime path must satisfy. 
\item We show how prime path generation can be reduced to cycle enumeration in an augmented graph constructed from the original input. 
\item We design an algorithm that incrementally generates prime paths without explicitly enumerating all simple paths. The algorithm integrates SCC-based decomposition, endpoint filtering, and early pruning, and uses a Johnson-style traversal as an enumeration mechanism within the reduced search space.
\item We identify SCC boundaries as checkpoints for termination, allowing us to discard some subtrees during further exploration.
\item We implement our approach, compare it with other approaches existing in the literature, and evaluate it on a large dataset of real-world control flow graphs. 
\item We demonstrate how our approach can be directly applied to generate test coverage items other than prime paths, such as simple paths or simple cycles.
\end{itemize}

\section{Preliminaries} \label{sec:preliminaries}

Let $G=(V, E)$ be a directed graph, where $V$ is a finite set of vertices and $E \subseteq V \times V$ is a finite set of edges. A {\it path} in $G$ is a sequence $(v_1, v_2, \dots, v_n)$ of vertices with $n \geq 1$ such that $(v_i, v_{i+1}) \in E$ for each $i = 1, \dots, n-1$. The {\it length} of a path $p = (v_1, \dots, v_n)$, denoted by $|p|$, is $n-1$, i.e., the number of edges in the path. Further, put $V((v_1, ..., v_n))=\{v_1, ..., v_n\}$. The set of all paths in $G$ is denoted by $P(G)$. A path $(v_1, \dots, v_n)$ is {\it simple} if all its vertices are distinct, i.e., $\forall i,j \in \{1,2,\ldots,n\},\ i \ne j \Rightarrow v_i \ne v_j.$ The set of all simple paths in $G$ is denoted by $SP(G)$. A {\it simple cycle} in $G=(V, E)$ is a path $(v_1, v_2, ..., v_n, v_1)$, $n \geq 1$, such that $(v_1, ..., v_n)$ is a simple path. The set of all simple cycles in $G$ is denoted by $SC(G)$. A simple path $(v_1, ..., v_n)$ is  {\it forward extendable} if there exists $v \in V$ such that $(v_1, ..., v_n, v) \in SP(G) \cup SC(G)$, {\it backward extendable} if there exists $v \in V$ such that $(v, v_1, ..., v_n) \in SP(G) \cup SC(G)$, {\it extendable} if it is either forward or backward extendable, and {\it non-extendable} if it is not extendable. A path $p$ in $G$ is a {\it prime path} if it is a simple cycle or a non-extendable simple path. The set of all prime paths in $G$ is denoted by $PP(G)$. A path $(v_1, v_2, ..., v_n)$ in $G=(V, E)$ is {\it e-acyclic} if all its edges are unique, i.e., $\forall i, j \in \{1, \ldots n-1\}, i \neq j$, we have $(v_i,v_{i+1}) \neq (v_j, v_{j+1})$. The set of all e-acyclic paths in $G$ is denoted by $eAP(G)$. 

A directed graph $G = (V,E)$ is said to be {\it strongly connected} if for every pair of vertices 
$u,v \in V$ there exists a path from $u$ to $v$. Let $G = (V, E)$ be a directed graph and let $U \subseteq V$. The {\it induced subgraph of} $G$ {\it on} $U$, denoted by $G[U]$, is the graph $G[U] = (U, \{\, (u,v) \in E \mid u,v \in U \,\}).$ Let $G = (V,E)$ be a directed graph. A {\it strongly connected component} (\textit{SCC}) of $G$ is a subgraph of $G$ induced by a subset $SC \subseteq V$ such that $G[SC]$ is strongly connected and for every $v \in V \setminus SC$ the subgraph $G[SC \cup \{v\}]$ is not strongly connected. Given $v \in V$, denote by $SC(v)$ an SCC of $G$ that contains $v$. For $v \in V$ we define its \textit{set of outgoing vertices} $out(v) = \{u \in V | (v,u) \in E\}$, and \textit{set of incoming vertices} $in(v) = \{u \in V | (u,v) \in E\}$.  A directed graph $G=(V,E)$ is called a {\it single entry single exit graph} ({\it SESE}), if it satisfies the following two conditions:

\begin{enumerate}
    \item\label{sese1} $|V|\geq 2$ and there exist distinct vertices $s(G), t(G) \in V$ (called resp. the entry and the exit vertex) such that $in(s(G))=\varnothing$ and $out(t(G))=\varnothing$.
    \item\label{sese2} Every vertex $v \in V$ belongs to at least one path from $s(G)$ to $t(G)$.
\end{enumerate} 

Given $G=(V,E)$, its {\it line graph} $L(G) = (L(V), L(E))$ is defined as follows: $L(V) = E$, $L(E) = \{ ((u, v),(v, w)) \mid (u, v) \in E \land (v, w) \in E \}$. The {\it line path reduction} is a function $pr \colon P(L(G))\to P(G)$ defined by $pr(((v_{1}, v_{2}), (v_{2}, v_{3}), \dots, (v_{n-1}, v_{n}))) = (v_{1}, v_2,\dots,v_{n}).$

The {\it rotations set of a simple cycle} $(v_1, ..., v_n, v_1)$ is a set $R((v_{1},v_{2},\dots,v_{n},v_{1}))=\\\left\{(v_i, v_{i+1}, \dots, v_n, v_1, v_2, \dots, v_i) \;|\; i = 1, 2, \dots, n\right\}.$
Let $G$ be an SESE graph that models the control flow of a program. A {\it test path} in $G$ is a path from $s(G)$ to $t(G)$. A test path $(v_1,\dots,v_n)$ in $G$ {\it covers} a path $(w_1,\dots,w_k)$ if there exists $i \in \{1, ..., n-k+1\}$ such that $(w_1, ..., w_k) = (v_i, ..., v_{i+k-1}).$ A set $TP$ of test paths covers a set of paths $S$ if for each $p \in S$ there exists $t \in TP$ such that $t$ covers $p$. 

\section{Prime Paths Characterization}\label{PPCharac}

This section establishes theoretical properties of prime paths in directed graphs, which are used in the following section to design a more efficient algorithm that enumerates all prime paths in a given graph. We start with preliminary facts that are folklore.
\begin{proposition}
In a directed graph $G = (V, E)$, every $v \in V$ belongs to exactly one SCC of $G$.
\end{proposition}
 
\begin{proposition}
 Let $G = (V, E)$ be a directed graph. There exists a unique partition $SC_1, \ldots, SC_m$ of $V$ where $SC_i$ induces an SCC in $G$ for every $1 \leq i \leq m$.
\end{proposition}

\begin{figure}[ht]
  \centering
  \begin{tikzpicture}[
    node/.style={circle, draw=black, fill=white, inner sep=1.5pt, minimum size=6pt},
    arr/.style={-{Stealth[length=6pt, width=5pt]}, semithick}
]

\node[node] (s) at (0,0) {};
\node[node] (v1) at (0,1.7) {};
\node[node] (v2) at (0,3.3) {};
\node[node] (v3) at (1.0,4.0) {};

\node[node] (v4) at (1.5,0.8) {};
\node[node] (v5) at (2.7,0) {};
\node[node] (v6) at (2.8,2.7) {};
\node[node] (v7) at (7.6,2.8) {};

\node[node] (v8) at (5.7,3.0) {};
\node[node] (v9) at (4.5,2.0) {};
\node[node] (v10) at (4.6,4.8) {}; 

\node[node] (v11) at (1.3,2.9) {};
\node[node] (v12) at (4.2,3.3) {};
\node[node] (v13) at (5.7,1.4) {};
\node[node] (t) at (6.8,4.3) {};


 \draw[arr] (s) to[bend left=8] (v1);
 \draw[arr] (s) to[bend right=8] (v4);
 \draw[arr] (v2) to[bend right=8] (v11);
 \draw[arr] (v11) to[bend right=8] (v3);
 \draw[arr] (v3) to[bend right=8] (v2);
 \draw[arr] (v3) to[bend left=10] (v10);
 \draw[arr] (v10) to[bend left=10] (t);
 \draw[arr] (v1) to[bend left=10] (v11);
 \draw[arr] (v11) to[bend left=10] (v6);
 \draw[arr] (v13) to[bend left=10] (v7);
 \draw[arr] (v7) to[bend left=10] (v13);
 \draw[arr] (v7) to[bend right=15] (t);
 \draw[arr] (v9) to[bend left=15] (v13);
 \draw[arr] (v8) to[bend right=15] (v7);
 \draw[arr] (v4) to[bend right=15] (v5);
 \draw[arr] (v9) to[bend right=15] (v6);
 \draw[arr] (v5) to[bend right=15] (v9);
 \draw[arr] (v6) to[bend right=15] (v5);
 \draw[arr] (v6) to[bend right=15] (v12);
 \draw[arr] (v12) to[bend right=15] (v6);
 \draw[arr] (v8) to[bend right=15] (v12);
 \draw[arr] (v12) to[bend right=15] (v8);

 \node[below=2pt] at (s) {s};
 \node[right=2pt] at (t) {t};
 \node[right=2pt] at (v1) {$v_1$};
 \node[left=2pt] at (v2) {$v_2$};
 \node[above=2pt] at (v3) {$v_3$};
\node[right=2pt] at (v4) {$v_4$};
\node[right=2pt] at (v5) {$v_5$};
\node[above=2pt] at (v6) {$v_6$};
\node[right=2pt] at (v7) {$v_7$};
\node[above=2pt] at (v8) {$v_8$};
\node[above=2pt] at (v9) {$v_9$};
\node[above=2pt] at (v10) {$v_{10}$};
\node[below=3pt] at (v11) {$v_{11}$};
\node[above=2pt] at (v12) {$v_{12}$};
\node[left=2pt] at (v13) {$v_{13}$};

\end{tikzpicture}
  \caption{A SESE digraph $G_{ex}$}
  \label{fig:example_graph}
\end{figure}
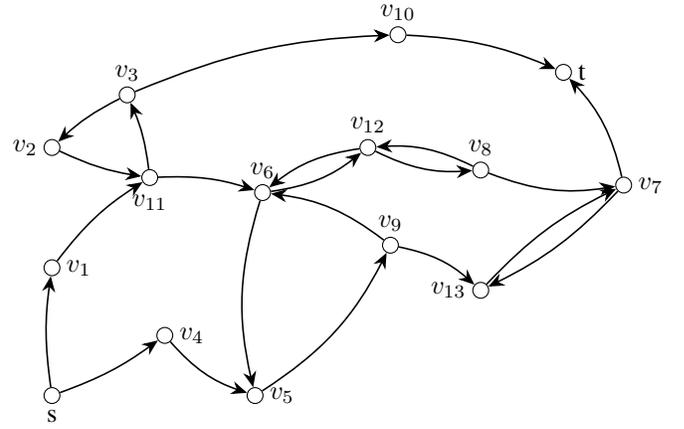

For example, the graph $G_{ex}$ from Fig. \ref{fig:example_graph} can be partitioned into SCCs induced by sets $\{s\}, \{v_1\}, \{v_4\}, \{v_2, v_3, v_{11}\}, \{v_{10}\}, \{v_5, v_6, v_8, v_9, v_{12}\}, \{v_7, v_{13}\},\ \{t\}$.

For a digraph $G = (V, E)$, we define a corresponding \textit{condensation graph} $Cond(G) = (\{SC_1, \ldots, SC_m\}, E_{Cond})$, where $SC_1, \ldots SC_m$ is a partition of $V$ inducing SCCs and \linebreak $(SC_i, SC_j) \in E_{Cond}$ iff there exist $v_i \in SC_i, v_j \in SC_j$, $i, j \in \{1, ..., m\}$, such that $(v_i, v_j) \in E$. It is easy to check the following two propositions.
\begin{proposition}
    For every digraph $G$, $Cond(G)$ is a directed acyclic graph. 
\end{proposition}
\begin{proposition}
\label{prop:simple_cycles}
If $c$ is a simple cycle in $G$, then all vertices of $c$ belong to one SCC of $G$.
\end{proposition}

We can now characterize prime paths in terms of SCCs of $G$. For a path $p = (v_1, \ldots, v_n)$ in a digraph $G$ with SCCs induced by $SC_1, \ldots, SC_m$, for each $1 \leq i \leq m$ we define a function $Cut_{SC_i}(p)$ that returns a set of vertices of $p$ contained in $SC_i$: $Cut_{SC_i}(p) = V(p) \cap SC_i.$

\begin{lemma}
\label{lemma:split_paths}
Let $G$ be a digraph, with SCCs induced by $SC_1, \ldots, SC_m$, and let $p = (v_1, \ldots, v_n)$ be a path in $G$. There exists a simple path $s = (SC_{t_1}, \dots, SC_{t_l})$ in $Cond(G)$ and a subsequence $(j_1, \ldots, j_l)$ of $(1,2, \ldots, n)$ with $j_1 = 1$, and $v_{j_i}, v_{j_i + 1} \ldots v_{j_{i+1} - 1} \in SC_{t_i}$ for all $i \in \{1, \ldots, l - 1\}$ and $v_{j_l}, \ldots, v_n \in SC_{t_l}$.
\end{lemma}

Note that for a sequence $(SC_{t_1}, \dots, SC_{t_l})$ given by Lemma \ref{lemma:split_paths} we have that $Cut_{SC_{t_i}}(p) \neq \varnothing$ for all $1 \leq i \leq l$ and that $\bigcup_{i=1}^l Cut_{SC_{t_i}}(p) = V(p)$.

\begin{theorem}
\label{thm:pp_characterization}
    Let $G$ be a digraph, with $SC_1, \ldots, SC_m$ inducing its SCCs and let $p = (v_1, \ldots, v_n)$ be a path in $G$. Then $p$ is a prime path if and only if one of the following conditions is satisfied: \begin{enumerate}
        \item \label{pp_characterization:cond_1} $p$ is a simple cycle;
        \item \label{pp_characterization:cond_2} $p$ is a simple path, there exists $k \in \mathbb{N}$ such that $V(p) \subseteq SC_k$, $in(v_1) \subseteq V(p) \setminus \{v_n\}$, $out(v_n) \subseteq V(p) \setminus \{v_1\}$;
        \item \label{pp_characterization:cond_3} there exists a simple path $(SC_{t_1}, \dots, SC_{t_l})$ in $Cond(G)$ with $Cut_{SC_{t_i}}(p) \neq \varnothing$ for all $1 \leq i \leq l$ and with $\bigcup_{i=1}^l Cut_{SC_{t_i}}(p) = V(p)$ such that $Cut_{SC_{t_i}}(p)$ induces a simple path with $in(v_1) \subseteq Cut_{SC_{t_1}}(p)$ and $out(v_k) \subseteq Cut_{SC_{t_l}}(p)$.
    \end{enumerate}
\end{theorem}

Examples of prime paths for $G_{ex}$ (see Fig. \ref{fig:example_graph}), corresponding to conditions of Theorem \ref{thm:pp_characterization} are as follows: $v_{11} \xrightarrow{} v_3 \xrightarrow{} v_2 \xrightarrow{} v_{11}$ (condition \ref{pp_characterization:cond_1}); $v_{9} \xrightarrow{} v_6 \xrightarrow{} v_5$ (condition \ref{pp_characterization:cond_2}); $v_{2} \xrightarrow{} v_{11} \xrightarrow{} v_3 \xrightarrow{} v_{10} \xrightarrow{} t$ (condition \ref{pp_characterization:cond_3}).

Given Theorem~\ref{thm:pp_characterization}, one approach to enumerating prime paths is to do it in a structured order: first, enumerate all simple cycles, and then simple non-extendable paths starting at a specified vertex 
$s$. However, the following result shows that such an approach is unlikely to be efficient in general. Recall that a polynomial delay algorithm returns the first output in polynomial time, and the time between any two consecutive outputs is also polynomial \cite{JOHNSON1988119}.
\begin{theorem}
\label{thm:np_poly_delay}
There does not exist an algorithm with polynomial delay that enumerates all simple non-extendable paths starting at a given vertex, unless $P=NP$.
\end{theorem}
This suggests that an efficient enumeration of prime paths in the desired order cannot be achieved efficiently in general. Therefore, instead of fully determining whether a vertex can be the beginning (or end) of a prime path, we focus on deriving necessary conditions that restrict the set of candidate vertices, thereby narrowing the search space in our algorithm. For a given $G=(V, E)$ we define the sets $V_{start}$ and $V_{end}$ as follows:

\begin{itemize}
    \item $v \in V_{start}$ if and only if the following four conditions hold:
    \begin{enumerate}[label=\textbf{S.\arabic*}]
    \item\label{lemma:prime_path_necessary:cond1} $in(v_1) \subseteq SC(v_1)$,
    \item \label{lemma:prime_path_necessary:cond3} $|\bigcup_{v \in in(v_1)} out(v) \cup \{v_1\}| > |in(v_1)|$, 
    \item\label{lemma:prime_path_necessary:cond5}  $|(\bigcup_{v \in in(v_1)} out(v)) \cap SC(v_1)| \geq |in(v_1)|$,
    \item\label{lemma:prime_path_necessary:cond7}  $|(\bigcup_{v \in in(v_1)} in(v)) \cap SC(v_1)| \geq |in(v_1)|$,
    \end{enumerate}
     \item $v \in V_{end}$ if and only if the following four conditions hold:
    \begin{enumerate}[label=\textbf{E.\arabic*}]
    \item \label{lemma:prime_path_necessary:cond2}  $out(v_n) \subseteq SC(v_n)$,
    \item\label{lemma:prime_path_necessary:cond4}  $|\bigcup_{v \in out(v_n) } in(v) \cup \{v_n\}| > |out(v_n)|$,
    \item\label{lemma:prime_path_necessary:cond6}  $|(\bigcup_{v \in out(v_n) } in(v)) \cap SC(v_n)| \geq |out(v_n)|$,
    \item\label{lemma:prime_path_necessary:cond8}  $|(\bigcup_{v \in out(v_n) } out(v)) \cap SC(v_n)| \geq |out(v_n)|$.
    \end{enumerate}
\end{itemize}

\ref{lemma:prime_path_necessary:cond1}-\ref{lemma:prime_path_necessary:cond7} (resp. \ref{lemma:prime_path_necessary:cond2}-\ref{lemma:prime_path_necessary:cond8}) are indeed necessary conditions for vertices to be beginnings (resp. endings) of prime paths, as the following lemma shows.

\begin{lemma}
\label{lemma:prime_path_necessary}
Let $G = (V,E)$ be a digraph and $p=(v_1, ..., v_n)$ a prime path in $G$, $v_1 \neq v_n$. Then $v_1 \in V_{start}$ and $v_n \in V_{end}$.
\end{lemma}

To illustrate the difference between conditions \ref{lemma:prime_path_necessary:cond3} and \ref{lemma:prime_path_necessary:cond5}, consider the graphs $G_{3,!5}$ and $G_{!3,5}$ from Figs \ref{fig:example_3_not_5} and \ref{fig:example_5_not_3}. 

\begin{figure}[ht]
  \centering
  \begin{tikzpicture}[
    node/.style={circle, draw=black, fill=white, inner sep=1.5pt, minimum size=6pt},
    arr/.style={-{Stealth[length=6pt, width=5pt]}, semithick}
]

\node[node] (v1) at (1,1) {};
\node[node] (x) at (3,0) {};
\node[node] (y) at (3,1) {};
\node[node] (z) at (3,2) {};

\node[node] (a) at (5,0) {};
\node[node] (b) at (6.5,1) {};
\node[node] (c) at (5,2) {};

\draw[arr] (c) to (b);
\draw[arr] (b) to (a);
\draw[arr] (a) to (c);
\draw[arr] (y) to (c);
\draw[arr] (x) to (a);
\draw[arr] (z) to[bend left=60] (b);
\draw[arr] (z) to[bend left=8] (v1);
\draw[arr] (x) to[bend left=8] (v1);
\draw[arr] (y) to[bend left=8] (v1);
\draw[arr] (v1) to[bend left=8] (x);
\draw[arr] (v1) to[bend left=8] (y);
\draw[arr] (v1) to[bend left=8] (z);

\node[left=2pt] at (v1) {$v_1$};
\node[below=2pt] at (x) {$x$};
\node[above=2pt] at (y) {$y$};
\node[above=2pt] at (z) {$z$};
\node[below=2pt] at (a) {$a$};
\node[right=2pt] at (b) {$b$};
\node[above=2pt] at (c) {$c$};

\end{tikzpicture}
  \caption{A digraph $G_{3,!5}$}
  \label{fig:example_3_not_5}
\end{figure}
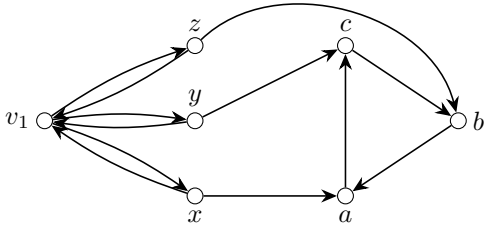

Note that $G_{3,!5}$ consists of two SCCs, $\{v_1, y,x,z\}$ and $\{a,b,c\}$. Also $in(v_1) = \{x,y,z\}$, $\bigcup_{v \in in(v_1)} out(v) \cup \{v_1\} = \{v_1, a,b,c\}$, so  \ref{lemma:prime_path_necessary:cond3} holds for $v_1$. But $(\bigcup_{v \in in(v_1)} out(v)) \cap SC(v_1) = \{v_1\}$ so  \ref{lemma:prime_path_necessary:cond5} does not hold for $v_1$.

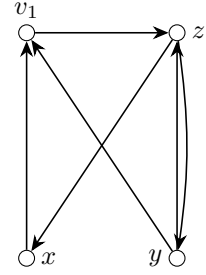
\begin{figure}[ht]
  \centering
  \begin{tikzpicture}[
    node/.style={circle, draw=black, fill=white, inner sep=1.5pt, minimum size=6pt},
    arr/.style={-{Stealth[length=6pt, width=5pt]}, semithick}
]

\node[node] (v1) at (0,4) {};
\node[node] (x) at (0,1) {};
\node[node] (y) at (2,1) {};
\node[node] (z) at (2,4) {};

\draw[arr] (x) to (v1);
\draw[arr] (y) to (v1);
\draw[arr] (y) to (z);
\draw[arr] (v1) to (z);
\draw[arr] (z) to[bend left=8] (y);
\draw[arr] (z) to (x);

 \node[above=2pt] at (v1) {$v_1$};
 \node[right=2pt] at (x) {$x$};
 \node[left=2pt] at (y) {$y$};
\node[right=2pt] at (z) {$z$};

\end{tikzpicture}
  \caption{A digraph $G_{!3,5}$}
  \label{fig:example_5_not_3}
\end{figure}

A digraph $G_{!3,5}$ is one SCC. Moreover $in(v_1) = \{x,y\}$, $\bigcup_{v \in in(v_1)} out(v) \cup \{v_1\} = \{v_1, z\}$, so  \ref{lemma:prime_path_necessary:cond3} does not hold for $v_1$. But $(\bigcup_{v \in in(v_1)} out(v)) \cap SC(v_1) = \{v_1, z\}$ so  \ref{lemma:prime_path_necessary:cond5} holds for $v_1$.

Similar examples (by reversing all edges of $G_{3,!5}$ and $G_{!3,5}$) can be given for the difference between \ref{lemma:prime_path_necessary:cond4} and \ref{lemma:prime_path_necessary:cond6}. The former shows that it may be important to include all the conditions from that lemma into our algorithm.

For $G=(V, E)$ and $v \in V_{start}(G)$ define $Ex(G,v) = (V \cup \{x_{v}\}, E')$, where $E' = E \cup \{(x_v, v) \} \cup \{(u, x_v) | u \in V_{end}(G) \wedge (u, v) \not\in E \wedge \exists (v, ..., u) \in P(G)\}$. Also, let $\mathcal{BE}_{G}$ be a predicate that, for a given path $p$, returns true if and only if $p$ is backward extendable in $G$.  

\begin{lemma}
\label{lemma:ex_g_cycles_to_g_paths}
Let $G=(V, E)$, $v_1 \in V_{start}$, $v_n \in V_{end}$. A path $(x_{v_1}, v_1, \ldots, v_n, x_{v_1})$ is a simple cycle in $Ex(G, v_1) = (V \cup \{x_{v_1}\}, E')$ if and only if $(v_1, \ldots, v_n)$ is a simple path in $G$ and $(v_n,v_1)\notin E$.
\end{lemma}

Also, observe the following property.

\begin{proposition}
\label{prop:forward_extend_with_ex_g}
    Let $(x_{v_1}, v_1, \ldots, v_n, x_{v_1})$ be a simple cycle in $Ex(G,v_1)$. A path $p =(v_1, \ldots, v_n)$ is backward extendable in $G$ if and only if there exists $v' \in SC(v_1) \cap in(v_1)$ such that $v' \notin V(p)$. 
\end{proposition}

\section{An approach based on Johnson's algorithm}\label{sec:Cycles}

\subsection{Overview of our approach}
Our goal is to efficiently enumerate all prime paths in a directed graph without explicitly generating the full set of simple paths, which is typically exponential in size. Our approach is based on three key observations:
\begin{itemize}
    \item \textbf{Structural decomposition via SCCs} -- every simple cycle is fully contained within a single SCC, while every non-cyclic prime path can be ``split'' using a sequence of SCCs that form a path in the condensation graph.
    \item \textbf{Constraints for being a prime path endpoint} -- using the necessary conditions (see Lemma \ref{lemma:prime_path_necessary}), we reduce the number of possible prime path endpoints to check. 
    \item \textbf{SCC boundary as a pruning checkpoint} -- when a path leaves an SCC in $G$ during the algorithm execution, a backward extendability check sometimes allows us to prune the subpaths search tree.
\end{itemize}

Based on these observations, we propose an algorithm for finding prime paths, inspired by Johnson's algorithm \cite{johnson1975}. However, our method does not enumerate all paths and filter them afterward. Instead, it performs constraint-driven cycle and path generation, discarding infeasible paths during construction.

\subsection{The algorithm}

Algorithm \ref{alg:ptime_path_generator} generates all prime paths for a given $G=(V, E)$. 

\begin{algorithm}[!ht]
\caption{\textsc{PrimePaths}}
\label{alg:ptime_path_generator}
\begin{algorithmic}[1]
\Require $G = (V, E)$
\Ensure  prime paths of $G$ (one by one)
\State \label{alg:ptime_path_generator:simple_cycles} $P = \textsc{SimpleCycles(G)}$ 
\For{$p\in P$}
    \For {$r \in R(p)$} \Comment{all rotations of $p$ (defined in Section \ref{sec:preliminaries})}
        \State \textbf{yield} $r$
    \EndFor
\EndFor
\State \label{alg:ptime_path_generator:non_extendable_paths} $P = \textsc{NonExtendableSimplePaths(G)}$ 
\For{$p\in P$}
    \State \textbf{yield} $p$
\EndFor
\end{algorithmic}
\end{algorithm}

It uses Algorithm \ref{alg:jonsons_algorithm_core}, which generates all simple cycles, and Algorithm \ref{alg:jonsons_algorithm_non_extendable}, which generates all non-extendable simple paths. 

\begin{algorithm}[!ht]
\caption{\textsc{SimpleCycles}}
\label{alg:jonsons_algorithm_core}
\begin{algorithmic}[1]
\Require $G = (V, E)$
\Ensure simple cycles of $G$ (one by one)

\While{$\exists v \in V$}
   \State $SCC_{v} \leftarrow \text{unique } SCC \in \textsc{Tarjan}(V,E) \text{\ such that\ }  v\in SCC$
   \State $(BlockedDict, BlockedSet, Path) \leftarrow$  (\{\ \}, $\varnothing$, [\ ]) \Comment empty structures
   \State \textbf{yield from} \textsc{Circuit($SCC_v$, $v$, $BlockedDict$, $BlockedSet$, $Path$, $\epsilon$, $\epsilon$)}
   \State $(V,E)  \leftarrow \text{a subgraph of $(V,E)$ induced by $V
   \setminus \{v\}$}$
\EndWhile
\end{algorithmic}
\end{algorithm}

\begin{algorithm}[!ht]
\caption{\textsc{NonExtendableSimplePaths}}
\label{alg:jonsons_algorithm_non_extendable}
\begin{algorithmic}[1]
\Require $G = (V, E)$
\Ensure non-extendable simple paths of $G$ (one by one)
\For{$v \in V_{start}(G)$}
    \State compute $Ex(G, v)$
    \State $G_{ext} \leftarrow $ SCC of $Ex(G,v)$ containing $x_v$
    \State $(BD, BS, Path)$ = (\{\ \}, $\varnothing$, [\ ])
    \State $P \leftarrow \textsc{Circuit($G_{ext}$, $x_v$, $BD$, $BS$, $Path$, $G$,$v$)}$
    \For{$p\in P$}
        \State remove $x_v$ from $p$
        \If{p is non-extendable in $G$}
            \State \textbf{yield} $p$
        \EndIf
    \EndFor
\EndFor
\end{algorithmic}
\end{algorithm}
 
To find simple cycles, the main loop of Algorithm \ref{alg:jonsons_algorithm_core} selects a vertex, finds the SCC containing that vertex (for example, using Tarjan's algorithm \cite{tarjan_r}), and runs on that SCC Algorithm \ref{alg:jonsons_circuit}, responsible for computing all simple cycles containing a given $v\in V$. Finally, the processed vertex can be removed from the graph, because all cycles containing it have already been found.

 After running Algorithm \ref{alg:jonsons_algorithm_core} and getting all the simple cycles of $G$, we can run Algorithm \ref{alg:jonsons_algorithm_non_extendable} to get all the non-extendable simple paths of $G$. For every $v \in V_{start}$, first we construct a graph $Ex(G, v)$ and then execute Algorithm \ref{alg:jonsons_circuit} for $x_v$. Next, we remove $x_v$ from each simple cycle found and filter out the forward extendable simple paths. 

The function \textsc{Circuit}, used by both above-mentioned algorithms, is shown in Algorithm \ref{alg:jonsons_circuit} (it uses the function \textsc{Unblock} listed in Algorithm \ref{alg:unblock}).

\begin{algorithm}[!ht]
\caption{\textsc{Circuit}}
\label{alg:jonsons_circuit}
\begin{algorithmic}[1]
\Require $G=(V, E)$, $v \in V$,  $BlockedDict$, $BlockedSet$, $Path$, $G' = (V',E')$ - might be $\epsilon$, $v' \in V'$ - might be $\epsilon$
\Ensure simple cycles (possibly such that its subpaths are not backward extendable) in $G$ containing $v$
\State $proceed \gets false$
\If{$v \notin V(G') \lor v' = \epsilon \lor SC_{G'}(v') = SC_{G'}(v)$}
\State $proceed \gets true$
\ElsIf{$(SC_{G'}(v') \neq SC_{G'}(v) \land !\mathcal{BE}_{G'}(Path[1:]))$} \label{alg:jonsons_circuit:modified}\Comment{see Prop. \ref{prop:forward_extend_with_ex_g}}
\State $proceed \gets true$
\State $G' \gets \epsilon, v' \gets \epsilon$
\EndIf
\If{$proceed$}
\State $found \leftarrow False$
\State append $v$ \text{to} $Path$
\State add $v$ to $BlockedSet$
\For{ $w\in V$ such that $(v, w)\in E$}
    \If{$w$ = head($Path$)} \Comment head($x$) = the first element of $x$
        \State \textbf{yield} $Path + w$ \Comment return a list representing a cycle from $w$ to $w$ 
        \State $found \gets True$ 
    \ElsIf{$w\not\in$ $BlockedSet$} 
        \State \textbf{yield} \textsc{Circuit($G$, $w$, $BlockedDict$, $BlockedSet$, $Path$, $G'$, $v'$)} 
        \State $found \leftarrow True$
    \EndIf
\EndFor
\If{$found$}
    \State $\textsc{Unblock}(v, BlockedSet, BlockedDict)$
\Else \label{alg:jonsons_circuit:blocking}
    \For{$w\in V$ such that $(v, w)\in E$}
        \If{$v \notin BlockedDict(w)$} \Comment $BlockedDict(w)$ is a list
            \State append $v$ to $BlockedDict(w)$  
        \EndIf
    \EndFor
\EndIf
\State remove the last element from $Path$
\EndIf
\end{algorithmic}
\end{algorithm}

The critical pruning step occurs at lines 2-6. When the current vertex $v$ belongs to a different SCC than the starting vertex $v'$, we apply Proposition \ref{prop:forward_extend_with_ex_g} to check whether the current partial path $Path[1:]$ (i.e., $Path$ with its first vertex removed) is already backward extendable in $G$. If it is the case, continuing this path cannot yield a prime path, so we stop the recursive calls.

\begin{lemma}
\label{lemma:johnson_not_modified}
Let $G = (V,E)$ be a digraph, $v \in V$ and $SCC_v$ be an SCC of $G$ containing $v$. Then $\textsc{Circuit($SCC_v$,$v$, $\{\}$,$[]$,$\varnothing$, $\epsilon$, $\epsilon$)}$ returns all simple cycles of $SCC_v$.
\end{lemma}


The computational complexity of Algorithm \ref{alg:ptime_path_generator} is the focus of the theorem below.

\begin{theorem}\label{thm:algorithm_correctness}
Let $G = (V, E)$ be a digraph, $c$ be the number of its simple cycles, and $s$ be the number of its simple paths.  Algorithm \ref{alg:ptime_path_generator} (\textsc{PrimePaths}) computes all prime paths for $G$ and runs in $O((|V|+|E|)(c+s))$.
\end{theorem}
Obviously, the number of prime paths can be exponential (as argued using $D_n$ graphs construction from Fig. \ref{fig:intro_example}), so the bound is nearly tight. Due to Theorem \ref{thm:np_poly_delay} the algorithm has no polynomial delay, but the amortized delay for the algorithm would be $O((|V|+|E|)\frac{(c+s)}{p})$ where $p$ is a number of prime paths, which, assuming $O(s) = p$ is a reasonable tradeoff. For DAGs, we have even: 

\begin{corollary}
    \label{cor:poly_delay_dags}
    If $G = (V, E)$ is acyclic, then  Algorithm \ref{alg:ptime_path_generator} has polynomial delay.
\end{corollary}

In the statistics of our dataset 
(Table~\ref{tab:cfg_stats}), $31.36\%$ of CFGs are 
DAGs, making this a particularly relevant property 
in practice.

\subsection{Comparison with previous approaches}
A well-known approach to computing prime paths is the Ammann and Offutt algorithm \cite{1Ammann2008} (Algorithm \ref{alg:ammann_offutt} in Appendix \ref{sec:app-c}), which serves as a baseline for many subsequent works. It can be easily optimized to reduce unnecessary computations. The procedure begins by enumerating the set $P$ of all non-forward-extendable simple paths in the graph. Next, the algorithm iteratively selects the longest paths from $P$ and removes all their sub-paths from $P$. While conceptually straightforward, this approach is computationally expensive. In particular, the need to repeatedly compare paths within $P$ leads to at least quadratic complexity with respect to $|P|$, which can be prohibitively large for realistic graphs.
In contrast, the improved approach (Algorithm \ref{alg:ammann_offutt_improved} in Appendix \ref{sec:app-c}) adopts a more efficient iterative strategy. Instead of maintaining the entire set of simple paths, it keeps only the set \textit{newExPaths}, which contains non-forward-extendable paths of a fixed length $k$ in the $k$-th iteration. At each step, the algorithm also identifies prime paths of length $k$ and removes them from further consideration. As a result, the search space is significantly reduced in every iteration. This eliminates the need for expensive global comparisons over all simple paths, making the algorithm substantially more efficient in practice.

A different algorithm is proposed in \cite{Fazli2019}. The key idea is to decompose the input graph into SCCs. Prime paths are first computed independently within each SCC, using a method similar to the Ammann-Offutt approach. Then, by constructing the condensation graph $\mathrm{Cond}(G)$, in which each SCC is treated as a single vertex, the locally computed prime paths are combined to form prime paths for the entire graph.

From a parallelization perspective, our approach offers notable advantages. Algorithms~\ref{alg:jonsons_algorithm_core} and \ref{alg:jonsons_algorithm_non_extendable} are naturally parallelizable. In particular, Algorithm~\ref{alg:jonsons_algorithm_non_extendable} can be parallelized by replacing the outer loop over starting vertices with a parallel loop. Since each prime path originates from a distinct starting vertex, the corresponding search spaces are disjoint, allowing for efficient parallel execution without synchronization overhead.
This stands in contrast to the approach from \cite{Fazli2019}, where parallelization is limited by the structure of SCCs. In practical control-flow graphs, SCCs tend to be small (with an average size of 7.33 nodes in our dataset), which restricts the available parallelism. Although a later modification of \cite{Fazli2019} introduces parallel computation within SCCs \cite{Fazli2023}, the benefits remain limited when SCCs are small or when the graph is acyclic. In particular, for DAGs, where SCCs are trivial, this approach offers little opportunity for parallel speedup.
By contrast, Algorithm~\ref{alg:jonsons_algorithm_core} allows parallel computation of cycles across different SCCs, while Algorithm~\ref{alg:jonsons_algorithm_non_extendable} enables parallel exploration from multiple starting vertices regardless of the SCC structure.

Unlike Fazli \cite{Fazli2019}, which uses SCCs for decomposition but completes path generation before applying extendability filters, our approach treats SCC boundaries as pruning checkpoints during traversals. This allows us to discard infeasible paths during path construction, rather than after their completion. In practice, this difference is significant -- for graphs with complex SCC structures, the pruned subtrees can be exponentially large.

\section{Experiments for Prime Paths Generation}\label{sec:Exp}

We conducted experiments to compare the prime path generation time with the Ammann and Offutt algorithm \cite{ammann2008} (Algorithm \ref{alg:ammann_offutt}), its improved version available at github.com/heshenghuan\-/Prime-Path-Coverage (Algorithm \ref{alg:ammann_offutt_improved}), and an algorithm proposed in \cite{Fazli2019}. We do not compare with \cite{Fazli2023} because it offers no new ideas beyond parallelization. The time comparison with \cite{Fazli2019} is incomplete, though, due to the number of timeouts in the \cite{Fazli2019} approach. All algorithms are implemented in Diblob package \cite{zelek2024} and are available in supplementary material. All experiments were performed on CFGs created from  C++ and Python code from GitHub. We selected repositories based on their popularity, choosing the top 1,000 projects from 2016-2025. For C++ functions, CFGs were generated from the compiled intermediate representation. Each project was compiled using the \texttt{clang} compiler~\cite{llvmclang}. Each CFG vertex in this dataset corresponds to a basic block in the source code. 

For Python functions, CFGs were derived directly from disassembled bytecode. We utilized the \texttt{bytecode} package, specifically its \texttt{ControlFlowGraph} class~\cite{bytecode}, to construct CFGs from Python functions. After \texttt{bytecode} transformation of Python code to the intermediate language, the CFGs were generated using that intermediate representation. As in the previous case, each CFG vertex corresponds to a single basic block of code. We present basic summary statistics of our dataset in Table \ref{tab:cfg_stats} and the scatterplot of the number of nodes versus the number of edges of graphs in our dataset in Fig. \ref{fig:plot}.

The summary statistics indicate that approximately 30\% of the analyzed graphs are DAGs. This observation is important because the algorithm in \cite{Fazli2019} is expected to perform only moderately well on such inputs. Moreover, the experimental evaluation in \cite{Fazli2019} is based on a limited dataset of only eight graphs drawn from a benchmark repository. In contrast, our experiments are conducted on a substantially larger and more diverse set of real-world CFGs extracted from actual software projects. Therefore, we argue that our evaluation provides a more realistic and representative assessment of algorithmic performance in practical scenarios. The dataset is a part of supplementary material. All experiments were run on a MacBook Air with an Apple M2 processor (8 cores) and 24 GB RAM.

The experimental results for CFGs from our dataset (Tables \ref{tab:results_not_many_prime_paths} -- \ref{tab:results_two_best}) show a clear and consistent advantage of our approach across all CFG size ranges. Timeouts are excluded from runtime statistics (mean/median), and we report their counts separately. Tables \ref{tab:results_small} -- \ref{tab:results_very_big} reveal that as graph size increases, competing methods fail to complete within the timeout on the majority of instances. For CFGs with 15-50 nodes, the Ammann-Offutt algorithm exceeds the timeout on 72\% of instances, and Fazli on 98\%, making the median speedup undefined for these methods. Our approach completes all instances with zero timeouts across every size category. For the largest graphs tested (201-1000 nodes), both competing methods time out on every instance, while our approach completes in under 3 minutes on average. We set the timeout to max($k \times our\_time$, $x$ min) per instance (different for each size of input graphs - for larger instances the $k$ was chosen as smaller due to significantly longer execution time). This adaptive criterion was chosen because both Fazli \cite{Fazli2019} and the original Ammann–Offutt algorithm \cite{ammann2008} exhibited runtimes several orders of magnitude larger than our approach on moderate instances, making a fixed global timeout either too restrictive (starving competitors on easy instances) or experimentally infeasible (hours of wall time per instance on larger ones). The $k$ multiplier ensures each competing method receives substantially more time than our approach requires, providing a conservative stopping criterion given the observed performance gap.

\subsection{Experiments for $D_n$ Graphs}
To evaluate the performance of all implemented algorithms in a simple yet challenging setting, where the number of prime paths grows exponentially, we conducted experiments on $D_n$ graphs (see Fig.~\ref{fig:intro_example}), which are DAGs. Additionally, we considered an extended version of $D_n$ in which the source and target vertices are connected, yielding a single SCC. The results of these experiments are presented in Tables~\ref{tab:results_diamond} and~\ref{tab:results_diamond_extended}. For our algorithm, the running times for the smallest instances may be affected by the initialization of the Diblob package. The experiments show that, in general, our approach significantly outperforms all the others in terms of both time and memory usage.

\subsection{Output delay}
We analyzed inter-output delays across 25 CFG instances, each containing between 1,000 and 2,000 nodes. The algorithm was initially run on 40 instances; however, cases with a small number of prime paths were excluded from the analysis. For each instance, the algorithm generated a maximum number of $5\cdot 10^7$ prime paths, and the time (delay) between the generation of each consecutive $10^5$ paths was measured. This gives us a maximum of $5 \cdot 10^2$ measurements for each instance, for which all the reported statistics were computed, as shown in Table \ref{tab:delay_stats}. In most instances, the delay distribution is tightly concentrated, with coefficients of variation typically below 0.15.
Moreover, high percentiles (e.g., 95th percentile) remain close to the mean, indicating the absence of heavy-tailed behavior in most cases.
While a small number of instances exhibit higher variability due to occasional outliers, these do not dominate the overall behavior.
Importantly, regression analysis reveals no meaningful increase in delay over time, with slopes consistently close to zero across all instances.
These results suggest that although a worst-case polynomial delay cannot be guaranteed (unless P = NP), our algorithm exhibits stable, well-behaved delay in practice, consistent with average-case polynomial delay.

\section{Coverage Criteria}\label{sec:coverCrit}

Algorithm \ref{alg:ptime_path_generator} for generating prime paths can serve as a basis for creating algorithms to generate sets of test paths that meet multiple path-based coverage criteria.  The advantage of our solution is that individual test paths are returned immediately upon detection. This is particularly useful when paths need to be transformed into test cases, because this can be done independently of generating the test paths. In Appendix \ref{sec:app-b}, we present algorithms for designing test paths that satisfy four types of path-related test coverage criteria. All algorithms accumulate test coverage items along the shortest path between the first and last coverage items they yield. If that is not possible, the algorithm finds the shortest path to the final vertex and yields a test path based on a parameter ($k$) that controls the maximum number of coverage items accumulated per test path. The implementation of the mentioned algorithms is available in the Diblob package \cite{zelek2024}.

In this paper, we consider the following four path-based coverage criteria: 
\begin{enumerate}
    \item Prime Path Coverage (requires $TP$ to cover $PP(G)$), 
    \item Simple Cycle Coverage (requires $TP$ to cover $SC(G)$),
    \item Simple Path Coverage (requires $TP$ to cover $SP(G)$),
    \item e-Acyclic Path Coverage (requires $TP$ to cover $eAP(G)$).
\end{enumerate}


\section{Conclusions}

We presented a novel, space- and time-efficient algorithm for generating prime paths. Our empirical evaluation on real-world CFGs demonstrates that the proposed algorithm consistently outperforms state-of-the-art approaches in both execution time and memory usage. Notably, these improvements hold even when our method is applied to substantially larger graphs than those considered in prior works. 

Several promising directions for future research emerge from this work. First, our framework could be extended to support additional coverage criteria, such as data flow-based. Second, integrating the algorithm into industrial testing tools and continuous integration pipelines would enable further validation of its practical impact in real development workflows. Third, future work could explore adaptive strategies that dynamically adjust path-generation limits based on observed fault detection effectiveness or resource constraints. Finally, parallelization may further improve performance on very large or highly complex control-flow graphs.

From a theoretical perspective, it would be interesting to ask about the FPT polynomial delay of the problem of enumerating all prime paths (a natural candidate for a parameter would be either the maximal degree of a vertex in a digraph or the maximal size of an SCC of a digraph). The second natural question to ask is whether there exists a prime path enumeration algorithm with polynomial delay without the assumption that the order of the outputs is specific, as stated in the proof of Theorem \ref{thm:np_poly_delay}.

\bibliography{references}
\newpage
\appendix

\section{Proofs of lemmata and theorems}

Proof of \textbf{Lemma~\ref{lemma:split_paths}}
\begin{proof}
We construct sequences $(t_1, \ldots t_l)$ and $(j_1, \ldots j_l)$ as follows. Put $J=(j_1)=(1)$ and $T=(t_1)$ such that $SC_{t_1}$ induces the unique SCC containing $v_1$. If $n=1$, it is easy to check that the lemma holds.
If $n>1$, we repeatedly modify $J$ and $T$, considering vertices $v_2$ to $v_n$ of $p$ in sequence. Suppose that the next vertex to consider is $v_i$ with $J=(j_1,\ldots,j_\alpha)$ and $T=(t_1,\ldots,t_\alpha)$ constructed so far. If $v_i \in SC_{t_\alpha}$, do not change either $J$ or $T$ and proceed to $v_{i+1}$. If, on the other hand, $v_i \notin SC_{t_\alpha}$, modify $J$ to $(j_1,\ldots,j_\alpha,j_{\alpha+1})$ with $j_{\alpha+1}=i$ and modify $T$ to $(t_1,\ldots,t_\alpha,t_{\alpha+1})$ such that $SC_{t_{\alpha+1}}$ induces the unique SCC containing $v_i$. Note that $t_{\alpha+1} \neq t_\beta$ for $\beta \in \{1,\ldots,\alpha\}$, since otherwise there would exist a cycle in $G$ with vertices in more than one SCC.
After repeating the procedure above for the rest of the vertices of $p$, it is easy to check that $J$ and $T$ are as desired.
\end{proof}

Proof of \textbf{Theorem~\ref{thm:pp_characterization}}
\begin{proof} 
``$\Rightarrow$''. Assume that $p = (v_1, \ldots, v_n)$ is a prime path in $G$. If $p$ is a simple cycle, then condition \ref{pp_characterization:cond_1} holds, so by definition, the only other possibility is that $p$ is a simple path that cannot be extended. Two cases remain, depending on whether $p$ is contained in a single SCC of $G$.

Case 1. $\exists_{k \in \mathbb{N}} V(p) \subseteq SC_k$. Since $p$ cannot be extended, then $v_1 \notin out(v_n)$ and $v_n \notin in(v_1)$ (otherwise $p$ could have been extended to a simple cycle). Furthermore, $in(v_1) \subseteq V(p)$ and $out(v_n) \subseteq V(p)$, since otherwise $p$ also could have been extended. Therefore condition \ref{pp_characterization:cond_2} holds.

Case 2. $\forall_{k \in \mathbb{N}} V(p) \not\subseteq SC_k$. From Lemma \ref{lemma:split_paths} we get that there exists a simple path $(SC_{t_1}, \dots, SC_{t_l})$ in $Cond(G)$ with $Cut_{SC_{t_i}}(p) \neq \varnothing$ for all $1 \leq i \leq l$ and with $\bigcup_{i=1}^l Cut_{SC_{t_i}}(p) = V(p)$. Since $p$ is a simple path, $Cut_{SC_{t_i}}(p)$ also induces a simple path for every $i \in \{1, \ldots, l\}$. It remains to prove that $in(v_1) \subseteq Cut_{SC_{t_1}}(p)$ and $out(v_n) \subseteq Cut_{SC_{t_l}}(p)$.

Suppose that $in(v_1) \nsubseteq Cut_{SC_{t_1}}(p)$. Then one of the following must hold: \begin{enumerate}
    \item $in(v_1) \cap ( SC_{t_1} \setminus Cut_{SC_{t_1}}(p)) \neq \varnothing$ -- but that would make $p$ extendable
    \item $in(v_1) \cap SC_{t} \neq \varnothing$ and $t \notin \{t_1, \ldots, t_l\}$ -- but that would make $p$ extendable
    \item $in(v_1) \cap SC_{t} \neq \varnothing$ and $t \in \{ t_2, \ldots, t_l\}$ -- but that would contradict the fact that $Cond(G)$ is a directed acyclic graph.
\end{enumerate}

The argument for the fact that $out(v_n) \subseteq Cut_{SC_{t_l}}(p)$ is analogous.

``$\Leftarrow$''. If condition \ref{pp_characterization:cond_1} holds, $p$ is a prime path by definition. If condition \ref{pp_characterization:cond_2} holds, $p$ is a non-extendable simple path, so again it is a prime path by definition.

Assume that condition \ref{pp_characterization:cond_3} holds. Since for every $i \in \{1, \ldots, l\}$ we have that $Cut_{SC_{t_i}}(p)$ is a simple path and $SC_{t_1}, ..., SC_{t_l}$ form a partition of $V(G)$,  $p$ is a simple path. Additionally, since $in(v_1) \subseteq SC_{t_1}(p)$ and $out(v_n) \subseteq SC_{t_l}(p)$, $p$ is non-extendable, because all the vertices that could extend $p$ are already in $p$.
\end{proof}


Proof of \textbf{Theorem~\ref{thm:np_poly_delay}}
\begin{proof}

Assume that we are given an instance of a directed fixed vertex Hamiltonian path problem \cite{DBLP:books/fm/GareyJ79} (i.e. for a given $(G,s)$ we are looking for a Hamiltonian path starting at $s$; the proof its NP-completeness is a straightforward reduction by adding a vertex $s'$ and an edge $(s',s)$). Suppose that we have a polynomial delay algorithm $\mathcal{E}$ that enumerates all non-extendable simple paths from a given vertex. We will show that we can, using $\mathcal{E}$ as an oracle, decide if $(G,s)$ has a Hamiltonian path. For a given instance of a directed fixed vertex Hamiltonian path problem $(G,s)$, construct a graph $G' = (V', E')$ such that $V' = V(G) \cup \{t\}$ $E' = E(G) \cup \{(u,t) : u \in V(G) \} \cup \{(u,s) : u \in V(G) \setminus \{s\}\}$. We prove the following claims (for a path $p=(v_1, ..., v_k)$ and vertex $v$ by $p.v$ we mean a path $(v_1, ..., v_k, v)$).

\begin{claim}
\label{claim:np_poly_1}
If $p$ is a Hamiltonian path from $s$ in $G$, then $p.t$ is non-extendable in $G'$.
\end{claim}
\begin{proof}
Observe that $V(p.t) = V(G) \cup \{t\} = V'$, so the only way for $p.t$ to be extendable would be to a simple cycle. Since $out(t) = \varnothing$, this is impossible.
\end{proof}

\begin{claim}
\label{claim:np_poly_2}
If $p$ is a non-extendable simple path in $G'$ and $p$ starts with $s$, then $p$ is a path from $s$ to $t$.
\end{claim}
\begin{proof}
Let $p = (s, v_1, \ldots, v_k)$ be a non-extendable simple path in $G'$. Suppose $v_k \neq t$. Then $(v_k,t) \in E'$ and $t \notin V(p)$, since $t$ can only be the last vertex of a path because $out(t) = \varnothing$. Thus $p$ is forward extendable, a contradiction.
\end{proof}

\begin{claim}
\label{claim:np_poly_3}
If $p$ is a non-extendable simple path from $s$ to $t$ in $G'$, then $p$ restricted to $V(G)$ is a Hamiltonian path in $G$ starting at $s$.
\end{claim}
\begin{proof}
Obviously, $p$ restricted to $V(G)$ is a path in $G$. Moreover, $p$ is non-extendable; in particular, it is not backward extendable. Since $s$ starts $p$ and $(u,s) \in E'$ for every $u \in V(G) \setminus \{s\}$ then $V(G) \setminus\{s\} \subseteq V(p)$, so $p$ restricted to $V(G)$ is Hamiltonian.
\end{proof}
 To complete the proof, we provide a decision algorithm for the directed fixed Hamiltonian path problem. It proceeds as follows: 
 \begin{itemize}
     \item Construct $(G',s)$ from the directed fixed Hamiltonian path problem instance $(G,s)$ like described above.
     \item Run the algorithm $\mathcal{E}$ on $(G',s)$. By polynomial delay, the first output appears after $d(|V'|,|E'|)$ steps.
     \item If at least one output appears after $d(|V'|,|E'|)$ steps, then by Claim \ref{claim:np_poly_2} it is a simple path of the form $p.t$. By Claim \ref{claim:np_poly_3} we deduce that $p$ is Hamiltonian, and return $\textbf{true}$.
     \item If no output appears after $d(|V'|, |E'|)$ steps, then $G'$ has no non-extendable paths from $s$, so return \textbf{false} (if $G$ had a Hamiltonian path $p$ from $s$, then by Claim \ref{claim:np_poly_1}, $p.t$ would be a non-extendable path in $G'$ from $s$, contradicting the assumption that no output appeared).
 \end{itemize}

\end{proof}

Proof of \textbf{Lemma~\ref{lemma:prime_path_necessary}}

\begin{proof}
Conditions \ref{lemma:prime_path_necessary:cond1} and \ref{lemma:prime_path_necessary:cond2} easily follow from Theorem \ref{thm:pp_characterization}, conditions \ref{pp_characterization:cond_2} and \ref{pp_characterization:cond_3}.

Note that again from Theorem \ref{thm:pp_characterization}, all vertices from $in(v_1)$ must be on $p$ (and also in $SC(v_1)$). Further notice that for any subset $P \subseteq V(p)\setminus\{v_1,v_n\}$ we have $|\bigcup_{v \in P} out(v)| \geq |P|$. Indeed, denote by $(i_1,\ldots,i_\alpha)$ an increasing sequence of indices such that $P=\{v_{i_1},\ldots,v_{i_\alpha}\}$. Then $\{v_{i_1+1},\ldots,v_{i_\alpha+1}\}\subseteq\bigcup_{v \in P} out(v)$. Also, if $v_1 \in \bigcup_{v \in P} out(v)$, then the inequality is strict and if $\{v_{i_1},\ldots,v_{i_\alpha}\}$ are in a single SCC of $G$, then at most one of $\{v_{i_1+1},\ldots,v_{i_\alpha+1}\}$ is not in that same SCC. Similarly, we get that $|\bigcup_{v \in P} in(v)| \geq |P|$ with the same stipulations. Conditions \ref{lemma:prime_path_necessary:cond3} to \ref{lemma:prime_path_necessary:cond7} and \ref{lemma:prime_path_necessary:cond4} to \ref{lemma:prime_path_necessary:cond8} follow mostly from a combination of these properties.

We start with \ref{lemma:prime_path_necessary:cond3}. If $in(v_1) = \varnothing$, then we are done as the left-hand side is positive. If $in(v_1) \neq \varnothing$ then $v_1 \in \bigcup_{v \in in(v_1)} out(v)$ and we are also done. Condition \ref{lemma:prime_path_necessary:cond4} can be proved by a similar argumentation.

To prove \ref{lemma:prime_path_necessary:cond5} note that if $in(v_1) = \varnothing$, then the inequality is trivial. If $in(v_1) \neq \varnothing$, then again $v_1 \in \bigcup_{v \in in(v_1)} out(v)$ and we ``lose'' at most one vertex from the SCC as shown above, so the weak inequality holds. Condition \ref{lemma:prime_path_necessary:cond6} follows from a similar argument.

Using very similar arguments we get \ref{lemma:prime_path_necessary:cond7} and \ref{lemma:prime_path_necessary:cond8}. 
\end{proof}

Proof of \textbf{Lemma~\ref{lemma:ex_g_cycles_to_g_paths}}

\begin{proof}``$\Rightarrow$''. Assume $c=(x_{v_1},v_1,\dots,v_n,x_{v_1})$ is a simple cycle in $Ex(G,v_1)$. Note that the edges $(x_{v_1},v_1)$ and $(v_n, x_{v_1})$ belong to $E'$. For each $i=1,\dots,n-1$, the edge $(v_i,v_{i+1}) \in E$ is in $c$. Hence, $(v_1,\dots,v_n)$ is a path in $G$. Since $c$ is simple, all its vertices are pairwise distinct. In particular, $v_1, ..., v_n$ are pairwise distinct as well as distinct from $x_{v_1}$. Thus $(v_1,\dots,v_n)$ is a simple path in $G$. Since $(v_n, x_{v_1}) \in E'$, from the definition of $Ex(G, v_1)$ it follows that $(v_n, v_1) \notin E$.

``$\Leftarrow$''. Assume $(v_1,\dots,v_n)$ is a simple path in $G$, $v_1 \in V_{start}$ and $v_n \in V_{end}$. By the definition of $Ex(G,v_1)$, $(x_{v_1},v_1)\in E'$ and $(v_n,x_{v_1})\in E'$. Since $(v_1,\dots,v_n)$ is a simple path in $G$, every edge $(v_i,v_{i+1})$ lies in $E \subseteq E'$. These edges form a closed path $c = (x_{v_1},v_1,\dots,v_n,x_{v_1})$ in $Ex(G,v_1)$. As $v_1, \ldots, v_n$ are pairwise distinct, no vertex besides $x_{v_1}$ appears in $c$ more than once. Thus, the path is a simple cycle.
\end{proof}

Proof of \textbf{Lemma~\ref{lemma:johnson_not_modified}}
\begin{proof}
    It is easy to see that if the last two parameters of \textsc{Circuit} are $\epsilon$, then \textsc{Circuit} never reaches line \ref{alg:jonsons_circuit:modified}, so it acts like a Johnson's algorithm and the result holds.
\end{proof}

Proof of \textbf{Theorem~\ref{thm:algorithm_correctness}}
\begin{proof}
Recall that a prime path is either a simple cycle or a non-extendable simple path. From Proposition \ref{prop:simple_cycles} we know that if $c_0$ is a simple cycle, then there exists a unique $i \in \{1, \ldots, m\}$ such that $V(c_0) \subseteq SCC_i$. Since we compute all simple cycles for all SCCs of $G$ in line 1 of Algorithm \ref{alg:ptime_path_generator} (see Lemma \ref{lemma:johnson_not_modified}), all simple cycles of $G$ will be generated by the algorithm. 

Assume now that $v \in V_{start}$ and $c_0$ is a simple cycle in $Ex(G,v)$ that contains $x_v$. From Lemma \ref{lemma:ex_g_cycles_to_g_paths} we know that $c \setminus \{ x_v \}$ is a simple path that starts with a vertex from $V_{start}$ and ends with a vertex from $V_{end}$. By Lemma \ref{lemma:prime_path_necessary} every possible start of a prime path must be in $V_{start}$ and every possible end of a prime path must be in $V_{end}$, so from the construction of $Ex(G,v)$ we know that all possible endpoints for a given $v \in V_{start}$ are checked. So, if $c_0 \setminus \{x_v\}$ is contained in a single SCC, then we only need to check condition \ref{pp_characterization:cond_2} from Theorem \ref{thm:pp_characterization} to make sure that $c_0 \setminus \{x_v\}$ is a prime path. Otherwise, we need to check condition \ref{pp_characterization:cond_3}. 

To prove the complexity bound, recall the property of Johnson's algorithm (see Lemma 3 of \cite{johnson1975}) that at most $O(|V| + |E|)$ time elapses between two consecutive calls of $\textsc{Circuit}$ (that is without checking the condition in line \ref{alg:jonsons_circuit:modified}, thus proceeding as a classical Johnson's algorithm). However, adding the condition in line \ref{alg:jonsons_circuit:modified} does not change this upper bound. Indeed, we check whether a vertex belongs to a given SCC (this can be done in constant time assuming a proper implementation) and whether a path is backward extendable (this can be done in $O(|V| + |E|)$). Observe that the check of backward extendability is done only once per path (when, during construction, we first change the SCC of $G$ we are in), so the complexity bound for $\textsc{Circuit}$ holds.

Observe that Algorithm \ref{alg:ptime_path_generator} calls Algorithm \ref{alg:jonsons_algorithm_core} in line \ref{alg:ptime_path_generator:simple_cycles}, which executes in $O((|V| + |E|)c)$ as it is Johnson's algorithm. Line \ref{alg:ptime_path_generator:non_extendable_paths} executes Algorithm \ref{alg:jonsons_algorithm_non_extendable}. It is easy to check that computing $V_{start}$ and $V_{end}$ can be done in $O((|V|+|E|)s)$. Let $V_{start}=\{v_1,\ldots,v_k\}$ and $s(v_i)$ ($i \in \{1,\ldots,k\}$) be a number of simple paths from $v_i$ to vertices of $V_{end}$ that are not in $in(v_i)$. In the main loop of Algorithm \ref{alg:jonsons_algorithm_non_extendable} we start with computing the graph $Ex(G,v)$ for a vertex $v$ in $V_{start}$ which can be done in $O(|V| + |E|)$ as we are only adding one vertex and at most $|V|$ edges to $G$ as well as searching for vertices reachable from $v$. Then $\textsc{Circuit}$ is called on $Ex(G,v)$, which executes in $O((|V| + |E|)c')$, where $c'$ is the number of simple cycles of $Ex(G,v)$ containing $x_v$ (or 1 if the number of such cycles is 0). A number of simple cycles in $Ex(G,v)$ containing $x_v$ is however just $s(v)$ (see Lemma \ref{lemma:ex_g_cycles_to_g_paths}) so the bound is $O((|V| + |E|)s(v))$. Lines 6-9 are executed in $O(|V| \cdot s(v))$ (for every path yielded by \textsc{Circuit} we check extendability in $O(|V|)$), which is bounded by $O((|V| + |E|)s(v))$. Therefore execution of one iteration of the main loop for a given $v\in V_{start}$ takes at most $O((|V| + |E|)s(v))$, so execution of Algorithm \ref{alg:jonsons_algorithm_non_extendable} takes $O((|V| + |E|)(s(v_1)+\ldots+s(v_k)))$ which is bounded by $O((|V| + |E|)s)$.

Combining this with $O((|V| + |E|)c)$ ends the proof.
\end{proof}

Proof of \textbf{Corollary~\ref{cor:poly_delay_dags}}
\begin{proof}
Since $G$ is acyclic, it has no simple cycles, so 
Algorithm~\ref{alg:jonsons_algorithm_core} outputs nothing and 
terminates in time polynomial in $|V|$ and $|E|$. It suffices to analyze 
Algorithm~\ref{alg:jonsons_algorithm_non_extendable}. Since $G$ is acyclic, every SCC of $G$ is a singleton. We start with a trivial characterization of prime paths in DAGs.

\begin{claim}
\label{claim:dag_prime_paths_1}
Let $G$ be a DAG and $p = (v_1, \ldots, v_n)$ a simple 
path in $G$. Then $p$ is a prime path if and only if 
$\mathrm{in}(v_1) = \varnothing$ and 
$\mathrm{out}(v_n) = \varnothing$.
\end{claim}
%

\begin{claim}
\label{claim:dag_prime_paths_2}
Let $G$ be a DAG. Then $v \in V_{start}(G)$ 
if and only if $in(v) = \varnothing$, and 
$v \in V_{end}(G)$ if and only if 
$out(v) = \varnothing$.
\end{claim}


These claims follow directly from definitions of prime path, $V_{start}$, $in(v)$, and $out(v)$. By Claims~\ref{claim:dag_prime_paths_1} 
and~\ref{claim:dag_prime_paths_2}, every 
$v \in V_{start}$ satisfies 
$in(v) = \varnothing$. Since $G$ is acyclic, every 
$v \in V_{start}$ is a start point of at least 
one path to some $t \in V_{end}$, hence 
producing at least one prime path. 

Since every SCC is a singleton, the backward extendability check at $V_{start}$ costs $O(1)$: $in(v) = \varnothing$, so there are no incoming edges to check. 
Similarly, the forward extendability check at 
$V_{end}$ costs $O(1)$ since 
$\mathrm{out}(t) = \varnothing$.  
Algorithm~\ref{alg:jonsons_circuit} runs as standard 
Johnson on $G_{ext}$, which is a single SCC 
by Lemma~\ref{lemma:ex_g_cycles_to_g_paths}. By 
Lemma~3 of~\cite{johnson1975} the delay is 
$O(|V|+|E|)$. \end{proof}

\section{Algorithms for test paths generation achieving different path-based coverage criteria} \label{sec:app-b}

\subsection{Simple Cycle Coverage} \label{sub:SCC}
Simple Cycle Coverage requires test paths to cover all simple cycles in a graph. For some graphs, the number of simple cycles can be relatively large, e.g., when a graph is a complete digraph. This is one reason returning results one by one makes the test paths design process more effective. It is worth noting that a simple cycle can begin in \textit{any} node. This means that, for example, in the case of an oriented triangle $v_1 \rightarrow v_2 \rightarrow v_3 \rightarrow v_1$, there are three simple cycles to cover: $(v_1, v_2, v_3, v_1)$, $(v_2, v_3, v_1, v_2)$, and $(v_3, v_1, v_2, v_3)$. Our algorithms, by default, cover only one of them. Therefore, to cover all simple cycles, it would be sufficient to duplicate the first simple cycle found to cover all its rotations. We implement this option with the $DoubleCycle$ flag.

In Algorithm \ref{alg:SCCcoverage}, the function $ShortestPath(x, y)$ returns the shortest path between vertices $x$ and $y$ (if it exists). It returns an empty list if $x=y$, and the $None$ value if there is no such path. As mentioned, the parameter $k$ allows us to accumulate a concrete number of simple cycles before returning a test path. In practice, the higher the $k$ value is, the fewer the number of test paths produced. However, the test paths will be longer, and we will also have to wait longer for yielding the next ones. 

\begin{algorithm}[!ht] 
\caption{\textsc{SimpleCycleCoverage}} \label{alg:SCCcoverage}
\begin{algorithmic}[1] 
\Require SESE Graph $G = (V, E, s, t)$, $DoubleCycle$, $k$
\Ensure test paths that achieve Simple Cycle Coverage (one by one) 
\State $(CyclesCounter, TestPath) \leftarrow (0, [\ ])$ 
\State $P \leftarrow  \textsc{SimpleCycles}(G)$ 
\For{$c \in P$}
    \If{$DoubleCycle$}
        \State $c \leftarrow c + tail(c)$ \Comment tail($c$) is $c$ without its first element
    \EndIf
    \State $CyclesCounter \leftarrow CyclesCounter + 1$
    \If{$CyclesCounter = 1$}
        \State $TestPath \leftarrow ShortestPath(s, head(c)) + c$ 
        \If{$k=1$}
            \State $TestPath \leftarrow ShortestPath(last(c), t)$ \Comment last($c$) is the last element of $c$
            \State \textbf{yield} $TestPath$
            \State ($CyclesCounter, TestPath) \leftarrow (0, [\ ])$
        \EndIf
    \Else
        \State $SPath = ShortestPath(last(TestPath), head(c))$
        \If{$SPath$ is $None$}
            \State $TestPath \leftarrow TestPath + ShortestPath(last(TestPath), t)$
            \State \textbf{yield} $TestPath$
            \State $(CyclesCounter, TestPath) \leftarrow (1, ShortestPath(S,head(c)) + c)$
        \Else
            \If{$CyclesCounter = k$}
                 \State $TestPath \leftarrow TestPath + SPath + c$
                 \State $TestPath \leftarrow TestPath + ShortestPath(last(c), t)$
                 \State \textbf{yield} $TestPath$
                 \State $(CyclesCounter, TestPath) \leftarrow (0, [\ ])$
            \Else
                    \State $TestPath \leftarrow TestPath + ShortestPath(last(TestPath), head(c))$
                    \State $TestPath \leftarrow TestPath + c$
            \EndIf
        
        \EndIf
    \EndIf
\EndFor
\If{$CyclesCounter \neq 0$}
    \State $TestPath \leftarrow TestPath + ShortestPath(last(TestPath), t)$
    \State \textbf{yield} $TestPath$
\EndIf
\end{algorithmic}
\end{algorithm}

\subsection{Prime Path Coverage}
Prime Path Coverage requires test paths to cover all prime paths in a graph. The Algorithm \ref{alg:PrPathcoverage} we propose separates prime paths into simple cycles and non-extendable simple paths. It simply combines algorithms from Sections \ref{sub:SCC} and \ref{sub:SPC}. In the \textsc{SimpleCycleCoverage} algorithm, we need to set the $DoubleCycle$ flag to $True$. The parameter $k$ has the same regulating function as in the previous cases.

\begin{algorithm}[!ht]
\caption{\textsc{PrimePathCoverage}}
\label{alg:PrPathcoverage}
\begin{algorithmic}[1]
\Require SESE $G = (V, E), k$
\Ensure test paths that achieve Prime Path Coverage (one by one) 
\State $SCC \leftarrow \textsc{SimpleCycleCoverage}(G,True,k)$ 
\State $SPC \leftarrow \textsc{SimplePathCoverage}(G, k)$ 

\For{$p \in SCC$}
    \State \textbf{yield} $p$
\EndFor
\For{$p \in SPC$}
    \State \textbf{yield} $p$
\EndFor
\end{algorithmic}
\end{algorithm}

\subsection{Simple Path Coverage} \label{sub:SPC}

Simple Path Coverage requires test paths to cover all simple paths in a graph. The algorithm \textsc{SimplePathCoverage}$(G, k)$ we propose is similar to Algorithm \ref{alg:SCCcoverage}; therefore, we do not provide its implementation. To design the test paths for this criterion, we use precisely the same strategy as in Section \ref{sub:SCC}, but there is no longer a need to use the $DoubleCycles$ flag. The parameter $k$ works similarly to the case of simple cycles, accumulating non-extendable simple paths.  

\subsection{e-Acyclic Path Coverage}

The e-Acyclic Path Coverage requires test paths to cover all edge-acyclic paths in a graph. This means that cycles can appear in these paths provided that only vertices are repeated. The Algorithm \ref{alg:eAccoverage} we propose is based on the fact that achieving the e-Acyclic Path Coverage for an SESE graph $G=(V, E, s, t)$ is equivalent to achieving the Simple Path Coverage for $L(G)$.

\begin{theorem}\label{thm:eacyclic_cov}
Let $G=(V, E, s, t)$ be an SESE graph. Then $eAP(G) = pr(SP(L(G)))$.
\end{theorem}
\begin{proof}
``$\supseteq$'' Consider $p\in SP(L(G))$ and the corresponding $pr(p) \in pr(SP(L(G)))$ ($pr$ is a bijection between $P(L(G))$ and $P(G)$). Because $p = ((v_{1},v_{2}),\dots, (v_{n-1}, v_{n}))$ has no repeated vertices, $pr(p) = (v_{1}, \dots, v_{n})$ has no repeated edges, hence $pr(p) \in eAP(G)$.

``$\subseteq$'' Let $p \in eAP(G)$, then $p=(v_{1}, v_{2},\dots, v_{n})$ is without repeated edges. Consequently, $pr^{-1}(p) = ((v_{1}, v_{2}),(v_{2}, v_{3}),\dots, (v_{n-1}, v_{n}))$ has no repeated vertices, so  
$pr^{-1}(p) \in SP(L(G))$, and hence $p \in pr(SP(L(G)))$.
\end{proof}

As a result, to meet the criterion, it suffices to compute the line graph $L(G)$ of the input graph $G$ and then run the algorithm for simple paths. Note that by using recursive graph linearization $L(L(\dots(L(G))\ldots))$, we can provide further criteria that disallow repeating paths of length $n=2, 3, ...$.

\begin{algorithm}[!ht]
\caption{\textsc{e-AcyclicPathCoverage}}
\label{alg:eAccoverage}
\begin{algorithmic}[1]
\Require SESE $G = (V, E), k$ 
\Ensure test paths that achieve e-Acyclic Path Coverage (one by one)
\State $SPC \leftarrow \textsc{SimplePathCoverage}(L(G), k)$ 
\For{$p \in SPC$}
    \State \textbf{yield} $pr(p)$
\EndFor
\end{algorithmic}
\end{algorithm}

\section{Auxiliary listings and tables}

\label{sec:app-c}
\begin{algorithm}[!ht]
\caption{\textsc{Unblock}}
\label{alg:unblock}
\begin{algorithmic}[1]
\Require $u,BlockedSet, BlockedDict$
    \State remove $u$ from $BlockedSet$
    \ForAll{$w \in BlockedDict(u)$}
        \State remove $w$ from $BlockedDict(u)$
        \If{$w \in BlockedSet$}
            \State \Call{Unblock}{$w$}
        \EndIf
    \EndFor
\end{algorithmic}
\end{algorithm}

\begin{algorithm}[H] 
\caption{Ammann-Offutt approach from \cite{ammann2008}} \label{alg:ammann_offutt} \begin{algorithmic}[1] \Require $G = (V, E)$ \Comment a directed graph \Ensure $PP(G)$ \Comment a set of prime paths in $G$ \State $P' \gets V$ \Comment all paths $p$ in $G$ such that $|p|=0$ \State $T \gets \emptyset$ \Comment an auxiliary set of nonextendable paths \State $PP(G) \gets \emptyset$ \While{$P' \neq \emptyset$} \State remove an arbitrary path $p$ from $P'$ \If{$p$ is a simple path that is forward extendable} \State add to $P'$ all $(p, v) \in SP(G) \cup SC(G), v \in V$ \Else \State add $p$ to $T$ \EndIf \EndWhile \While{$T \neq \emptyset$} \State remove the longest path $p$ from $T$ \State add $p$ to $PP(G)$ \State remove from $T$ all subpaths of $p$ \EndWhile \State \Return $PP(G)$ \end{algorithmic} 
\end{algorithm}

\begin{algorithm}[H]
\caption{Improved Ammann-Offutt approach from \hyperlink{}{https://github.com/heshenghuan/Prime-Path-Coverage}}
\label{alg:ammann_offutt_improved}
\begin{algorithmic}[1]

\Require $G = (V, E)$ \Comment a directed graph
\Ensure $PP(G)$ \Comment the set of prime paths in $G$

\State $exPaths \gets \{ (v) \mid v \in V \}$ \Comment initial paths of length 1
\State $paths \gets \emptyset$

\Function{FindSimplePaths}{$G, exPaths, paths$}

    \ForAll{$p \in exPaths$}
        \If{$p$ is a prime path}
            \State add $p$ to $paths$
        \EndIf
    \EndFor

    \State $exPaths \gets \{ p \in exPaths \mid p$ is extendable in $G \}$

    \State $newExPaths \gets \emptyset$

    \ForAll{$p \in exPaths$}
        \ForAll{$v \in E(last(p))$}
            \If{$v \notin p$ \textbf{or} $v = first(p)$}
                \State add $(p, v)$ to $newExPaths$
            \EndIf
        \EndFor
    \EndFor

    \If{$newExPaths \neq \emptyset$}
        \State \Call{FindSimplePaths}{$G, newExPaths, paths$}
    \EndIf

    \State \Return $paths$
\EndFunction

\vspace{0.5em}

\State $SP \gets$ \Call{FindSimplePaths}{$G, exPaths, paths$}
\State $PP(G) \gets SP$

\State \Return $PP(G)$

\end{algorithmic}
\end{algorithm}

\begin{table}[ht]
\centering
\caption{Summary statistics of CFGs.}
\label{tab:cfg_stats}
\begin{tabular}{l r}
\hline
Characteristic & Value\\
\hline
Number of functions & 120,314 \\
DAG ratio & 31.36\% \\
\midrule
Average nodes per function & 30.46 \\
Average edges per function & 42.87 \\
Average SCCs per function & 21.88 \\
Average largest SCC size & 7.33 \\
Average degree & 1.37 \\
Average max in-degree & 5.63 \\
Average max out-degree & 2.59 \\
\midrule
Standard deviation of nodes & 55.27 \\
Standard deviation of edges & 77.27 \\
Standard deviation of SCCs per function & 42.27 \\
Standard deviation of largest SCC size & 19.57 \\
Standard deviation of degree & 0.20 \\
Standard deviation of max in-degree & 18.27 \\
Standard deviation of max out-degree & 1.08 \\
\midrule
Maximum nodes in a function & 5,402 \\
Maximum edges in a function & 8,441 \\
Maximum in-degree overall & 1,967 \\
Maximum out-degree overall & 15 \\
\midrule
Node count (25\% quantile) & 14 \\
Node count (50\% quantile / median) & 20 \\
Node count (75\% quantile) & 32 \\
Node count (90\% quantile) & 53 \\
Node count (99\% quantile) & 168 \\
\hline
\end{tabular}
\end{table}

\begin{figure}[!h]
    \centering
    \includegraphics[width=0.9\linewidth]{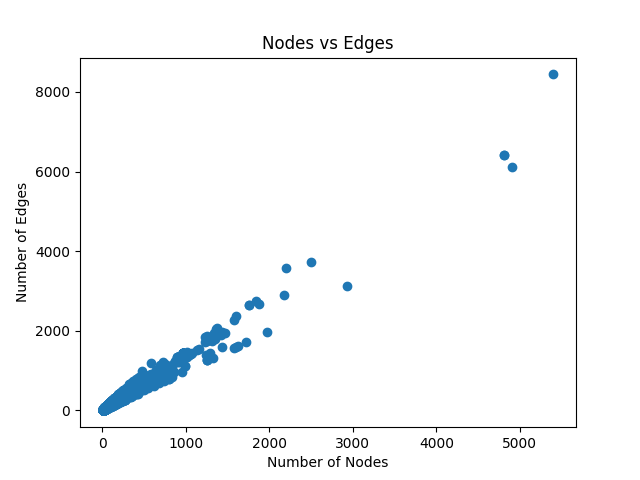}    
    \caption{The plot of nodes vs edges of the dataset}\label{fig:plot}
\end{figure}

\begin{table}[ht]
\centering
\caption{Summary of experiments for CFGs with less than 5000 prime paths (3629 instances), timeout = max(16*our\_time, 10min)}
\label{tab:results_not_many_prime_paths}
\begin{tabular}{lccccccc}
\toprule
& \multicolumn{2}{c}{Time(s)} & \multicolumn{2}{c}{Memory(MB)} & \multicolumn{2}{c}{Speedup} & Timeouts \\Method & mean & median & mean & median & mean & median & count \\
\midrule
Ammann-Offutt \cite{ammann2008} & 3.14 & 0.02 & 0.16 & 0.03 & - & - & 0 \\
Impr. Ammann-Offutt & 0.16 & \textbf{0.02} & 0.53 & \textbf{0.04} & 3.80 & 1.50 & 0 \\
Fazli \cite{Fazli2019} & 11.00 & 0.22 & 0.51 & 0.13 & 0.20 & 0.10 & 83 \\
Our approach & \textbf{0.15} & 0.04 & \textbf{0.15} & 0.13 & 11.63 & 0.68 & 0 \\
\bottomrule
\end{tabular}
\end{table}

\begin{table}[ht]
\centering
\caption{Summary of experiments for CFGs between 15 and 50 nodes (159 instances), timeout = max(16*our\_time, 10min)}
\label{tab:results_small}
\begin{tabular}{lccccccc}
\toprule
& \multicolumn{2}{c}{Time(s)} & \multicolumn{2}{c}{Memory(MB)} & \multicolumn{2}{c}{Speedup} & Timeouts \\Method & mean & median & mean & median & mean & median & count \\
\midrule
Ammann-Offutt \cite{ammann2008} & 311.24 & 252.12 & 4.19 & 4.16 & - & - & 114 \\
Impr. Ammann-Offutt & 42.03 & 10.63 & 234.14 & 59.50 & 31.72 & - & 2 \\
Fazli \cite{Fazli2019} & 90.89 & 81.65 & 3.98 & 4.40 & 7.28 & 7.15 & 156 \\
Our approach & \textbf{5.20} & \textbf{1.17} & \textbf{0.21} & \textbf{0.17} & 188.95 & - & 0 \\
\bottomrule
\end{tabular}
\end{table}

\begin{table}[ht]
\centering
\caption{Summary of experiments for CFGs between 51 and 100 nodes (41 instances), timeout = max(16*our\_time, 30min)}
\label{tab:results_medium}
\begin{tabular}{lccccccc}
\toprule
& \multicolumn{2}{c}{Time(s)} & \multicolumn{2}{c}{Memory(MB)} & \multicolumn{2}{c}{Speedup} & Timeout \\Method & mean & median & mean & median & mean & median & count \\
\midrule
Ammann-Offutt \cite{ammann2008} & 805.74 & 620.18 & 8.76 & 8.53 & - & - & 27 \\
Impr. Ammann-Offutt & 211.70 & 19.23 & 1003.82 & 86.35 & 40.59 & - & 3 \\
Fazli \cite{Fazli2019} & 577.83 & 451.55 & 89.45 & 11.74 & 2.25 & 0.44 & 34 \\
Our approach & \textbf{76.02} & \textbf{3.52} & \textbf{0.33} & \textbf{0.32} & 258.31 & - & 0 \\
\bottomrule
\end{tabular}
\end{table}

\begin{table}[ht]
\centering
\caption{Summary of experiments for CFGs between 101 and 200 nodes (19 instances), timeout = max(8*our\_time, 1h)}
\label{tab:results_big}
\begin{tabular}{lccccccc}
\toprule
& \multicolumn{2}{c}{Time(s)} & \multicolumn{2}{c}{Memory(MB)} & \multicolumn{2}{c}{Speedup} & Timeouts \\Method & mean & median & mean & median & mean & median & count \\
\midrule
Ammann-Offutt \cite{ammann2008} & 771.25 & 457.91 & 8.48 & 6.94 & - & - & 16 \\
Impr. Ammann-Offutt & 476.65 & 106.14 & 1563.01 & 314.89 & 16.23 & - & 6 \\
Fazli \cite{Fazli2019} & 853.78 & 433.62 & 58.34 & 13.77 & 0.66 & 0.93 & 16 \\
Our approach & \textbf{279.25} & \textbf{21.45} & \textbf{0.48} & \textbf{0.47} & 99.01 & - & 0 \\
\bottomrule
\end{tabular}
\end{table}

\begin{table}[ht]
\centering
\caption{Summary of experiments for CFGs between 201 and 1000 nodes (7 instances), timeout = max(4*our\_time, 2h)}
\label{tab:results_very_big}
\begin{tabular}{lccccccc}
\toprule
& \multicolumn{2}{c}{Time(s)} & \multicolumn{2}{c}{Memory(MB)} & \multicolumn{2}{c}{Speedup} & Timeouts \\Method & mean & median & mean & median & mean & median & count \\
\midrule
Ammann-Offutt \cite{ammann2008}  &-&- & -& - & - & - & 7 \\
Impr. Ammann-Offutt & 2634.56 & 843.12 & 3747.16 & 2535.23 & - & - & 0 \\
Fazli \cite{Fazli2019} &-&- & -& - & - & - & 7 \\
Our approach & \textbf{157.80} & \textbf{38.66} & \textbf{1.86} & \textbf{0.85} &-& - & 0 \\
\bottomrule
\end{tabular}
\end{table}

\begin{table}[ht]
\centering
\caption{Summary of experiments for 2 best algorithms for CFGs between 15 and 150 nodes  (555 instances), timeout = max(3*our\_time, 10m)}
\label{tab:results_two_best}
\begin{tabular}{lccccccc}
\toprule
& \multicolumn{2}{c}{Time(s)} & \multicolumn{2}{c}{Memory(MB)} & \multicolumn{2}{c}{Speedup} & Timeouts \\Method & mean & median & mean & median & mean & median & count \\
\midrule
Impr. Ammann-Offutt & 73.37 & 19.45 & 357.04 & 88.72 & - & - & 68 \\
Our approach & \textbf{26.75} & \textbf{3.12} & \textbf{0.30} & \textbf{0.30} & - & - & 0 \\
\bottomrule
\end{tabular}
\end{table}

\begin{table}[ht]
\centering
\caption{Experiments for $D_n$ graphs, timeout 5h}
\label{tab:results_diamond}
\begin{tabular}{c cc cc cc cc}
\toprule
& \multicolumn{2}{c}{Our approach} 
& \multicolumn{2}{c}{Impr. Ammann-Offutt} 
& \multicolumn{2}{c}{Fazli} 
& \multicolumn{2}{c}{Ammann-Offutt} \\

Size 
& time & mem 
& time & mem 
& time & mem 
& time & mem \\
\midrule

1  & 0.0027 & 0.0400 & \textbf{0.0010} & \textbf{0.0148} & 0.0022 & 0.0212 & \textbf{0.0010} & \textbf{0.0148} \\
2  & 0.0080 & 0.0533 & 0.0115 & 0.0155 & \textbf{0.0044} & 0.0389 & 0.0052 & \textbf{0.0150} \\
3  & 0.0042 & 0.0622 & 0.0023 & 0.0193 & 0.0107 & \textbf{0.0100} & \textbf{0.0021} & \textbf{0.0100} \\
4  & 0.0187 & \textbf{0.0100} & 0.0041 & 0.0289 & 0.0317 & 0.0902 & \textbf{0.0029} & 0.0184 \\
5  & \textbf{0.0065} & 0.0797 & 0.0078 & 0.0536 & 0.1435 & 0.1256 & 0.0080 & \textbf{0.0259} \\
6  & \textbf{0.0082} & 0.0973 & 0.0433 & 0.1154 & 0.6822 & 0.2397 & 0.0322 & \textbf{0.0437} \\
7  & \textbf{0.0334} & 0.1085 & 0.1121 & 0.2578 & 3.4360 & 0.4732 & 0.1200 & \textbf{0.0872} \\
8  & \textbf{0.0174} & \textbf{0.1160} & 0.0865 & 0.5707 & 10.2762 & 0.9501 & 0.4742 & 0.1790 \\
9  & \textbf{0.0967} & \textbf{0.1198} & 0.1820 & 1.2633 & 42.7592 & 1.9305 & 1.6172 & 0.3857 \\
10 & \textbf{0.0515} & \textbf{0.1214} & 0.4289 & 2.7821 & 161.0710 & 3.9514 & 6.0426 & 0.8055 \\
11 & \textbf{0.2872} & \textbf{0.1240} & 1.1380 & 6.0425 & 782.5986 & 8.1002 & 33.8452 & 1.7276 \\
12 & \textbf{0.1996} & \textbf{0.1373} & 2.4169 & 13.0652 & 3364.0473 & 16.6407 & 134.5254 & 3.7125 \\
13 & \textbf{0.4726} & \textbf{0.1432} & 4.6152 & 28.1320 & 15128.6997 & 34.2088 & 500.9064 & 7.9184 \\
14 & \textbf{0.7060} & \textbf{0.1579} & 9.6191 & 60.2659 & -- & 70.2951 & 2249.0866 & 16.8377 \\
15 & \textbf{1.8509} & \textbf{0.1727} & 20.3505 & 128.5302 & -- & 144.3377 & 9864.6615 & 35.6894 \\

\bottomrule
\end{tabular}
\end{table}

\begin{table}[ht]
\centering
\caption{Experiments for extended $D_n$ graphs, timeout 5h}
\label{tab:results_diamond_extended}
\begin{tabular}{c cc cc cc cc}
\toprule
& \multicolumn{2}{c}{Our approach} 
& \multicolumn{2}{c}{Impr. Ammann-Offutt} 
& \multicolumn{2}{c}{Fazli} 
& \multicolumn{2}{c}{Ammann-Offutt} \\

Size 
& time & mem 
& time & mem 
& time & mem 
& time & mem \\
\midrule

1  & 0.0178 & \textbf{0.0100} & 0.0315 & 0.0160 & 0.0160 & 0.0384 & \textbf{0.0037} & 0.0152 \\
2  & 0.0208 & \textbf{0.0100} & \textbf{0.0027} & \textbf{0.0100} & 0.0534 & \textbf{0.0100} & 0.0123 & \textbf{0.0100} \\
3  & 0.0190 & 0.0960 & \textbf{0.0057} & 0.0406 & 0.3722 & 0.1977 & 0.0144 & \textbf{0.0263} \\
4  & 0.0126 & 0.1102 & \textbf{0.0096} & \textbf{0.0100} & 0.3374 & 0.2010 & 0.0441 & \textbf{0.0100} \\
5  & 0.0607 & 0.1274 & \textbf{0.0125} & 0.0903 & 0.3240 & 0.2061 & 0.0577 & \textbf{0.0457} \\
6  & 0.0369 & 0.1587 & \textbf{0.0349} & 0.1725 & 0.3364 & 0.2100 & 0.0721 & \textbf{0.0681} \\
7  & \textbf{0.0625} & \textbf{0.1424} & 0.0987 & 0.3653 & 0.3807 & 0.2207 & 0.2587 & 0.1216 \\
8  & \textbf{0.0402} & \textbf{0.1583} & 0.1210 & 0.7615 & 0.5023 & 0.2592 & 0.9374 & 0.2520 \\
9  & \textbf{0.0501} & \textbf{0.1680} & 0.4500 & 1.6550 & 1.4077 & 0.4725 & 3.0341 & 0.4984 \\
10 & \textbf{0.0696} & \textbf{0.1778} & 0.6884 & 3.6028 & 3.9397 & 0.7988 & 9.7653 & 1.0485 \\
11 & \textbf{0.1956} & \textbf{0.1939} & 1.8770 & 7.7850 & 15.9820 & 1.7100 & 34.8094 & 2.2267 \\
12 & \textbf{0.2461} & \textbf{0.2007} & 3.6052 & 16.7890 & 55.7572 & 3.3379 & 151.6572 & 4.7626 \\
13 & \textbf{0.3914} & \textbf{0.2189} & 7.5845 & 36.0381 & 324.9300 & 6.8299 & 681.4616 & 10.1249 \\
14 & \textbf{0.8561} & \textbf{0.2303} & 11.4642 & 77.0805 & 1119.4746 & 13.6114 & 2963.1393 & 21.5068 \\
15 & \textbf{2.7371} & \textbf{0.2223} & 30.2840 & 164.1732 & 5253.7572 & 28.2040 & 13626.8657 & 45.5289 \\
16 & \textbf{3.6339} & \textbf{0.2322} & 55.8439 & 348.4011 & -- & 41.4446 & -- & 96.1081 \\

\bottomrule
\end{tabular}
\end{table}

\begin{table}[t]
\caption{Per-instance statistics of inter-output delays across benchmark functions. Mean, coefficient of variation (CV), median (p50), 95th percentile (p95), maximum delay (Max), and linear trend of delay over time are reported.}
\label{tab:delay_stats}
\centering
\small
\begin{tabular}{lrrrrrr}
\hline
Instance & Mean & CV & p50 & p95 & Max & Trend slope \\
\hline
ggml\_metal\_encode\_node & 3.83 & 0.087 & 3.76 & 4.39 & 6.29 & 0.00023 \\
ggml\_metal\_init (1) & 9.05 & 0.138 & 9.02 & 10.73 & 16.18 & -0.00104 \\
catch\_test\_84 & 14.11 & 0.137 & 13.88 & 16.05 & 37.11 & 0.00268 \\
catch\_test\_204 & 12.24 & 0.140 & 12.54 & 14.39 & 23.54 & -0.00640 \\
doctest\_anon\_func & 37.01 & 0.123 & 35.77 & 45.45 & 68.02 & -0.00916 \\
model\_writer\_save & 0.22 & 0.464 & 0.20 & 0.28 & 2.16 & -0.00008 \\
show\_demo\_widgets (1) & 13.65 & 0.149 & 13.19 & 16.54 & 30.51 & -0.00153 \\
ggml\_metal\_init (2) & 9.58 & 0.084 & 9.43 & 11.12 & 15.60 & -0.00180 \\
ggml\_metal\_encode\_node (2) & 1.99 & 0.097 & 1.97 & 2.26 & 3.98 & 0.00010 \\
show\_demo\_widgets (2) & 8.25 & 0.085 & 8.11 & 9.55 & 13.26 & 0.00181 \\
hashes\_test\_hscan & 9.22 & 0.079 & 9.05 & 10.67 & 15.13 & 0.00041 \\
keys\_scan\_all & 10.40 & 0.073 & 10.24 & 11.90 & 14.34 & 0.00060 \\
keys\_pattern\_match & 9.62 & 0.143 & 9.21 & 11.82 & 20.02 & 0.00243 \\
lists\_lrem\_test & 9.62 & 0.141 & 9.26 & 11.97 & 18.89 & -0.00298 \\
lists\_rpoplpush & 9.20 & 0.105 & 8.96 & 10.88 & 15.92 & 0.00113 \\
sets\_sscan & 8.92 & 0.080 & 8.78 & 10.23 & 13.32 & 0.00123 \\
zsets\_zscan & 9.28 & 0.073 & 9.17 & 10.36 & 15.05 & 0.00013 \\
mime\_infer & 3.85 & 0.143 & 3.75 & 4.12 & 7.23 & -0.00823 \\
xroonode\_draw & 1.40 & 1.219 & 0.19 & 3.97 & 7.96 & 0.00881 \\
core\_smooth\_test & 7.55 & 0.096 & 7.40 & 8.77 & 14.48 & -0.00011 \\
ast\_builder & 1.94 & 0.613 & 1.75 & 1.94 & 9.72 & -0.02002 \\
regex\_test & 9.53 & 0.080 & 9.39 & 10.89 & 14.99 & 0.00068 \\
aes\_dec & 0.36 & 1.080 & 0.32 & 0.47 & 8.41 & -0.00021 \\
aes\_enc & 0.35 & 1.104 & 0.32 & 0.42 & 8.81 & -0.00014 \\
show\_demo\_widgets (3) & 13.52 & 0.259 & 12.55 & 20.51 & 35.94 & 0.00606 \\
\hline
\end{tabular}
\end{table}

\end{document}